\documentclass[11pt,draftclsnofoot,peerreviewca,onecolumn,twoside]{IEEEtran}

\usepackage[cmex10]{amsmath}
\usepackage{amsfonts}	
\usepackage{cite}
\usepackage{amssymb}
\usepackage[T1]{fontenc}
\usepackage{graphicx}
\usepackage{color, subcaption}
\usepackage{amsthm}

\newcommand{\gt}{>}

\newcommand{\C}{\mathbb{C}}
\newcommand{\N}{\mathbb{N}}

\newcommand{\Z}{\mathbb{Z}}

\newcommand{\we}{\tilde{\mathbf{w}}}
\newcommand{\wei}{\tilde{\mathbf{w}}_i}
\newcommand{\weii}{\tilde{\mathbf{w}}_{i+1}}
\newcommand{\weim}{\tilde{\mathbf{w}}_{i-1}}

\newcommand{\w}{\mathbf{w}}
\newcommand{\e}{\mathbf{e}}
\newcommand{\ve}{\mathbf{v}}
\newcommand{\X}{\mathbf{X}}
\newcommand{\x}{\mathbf{x}}

\newcommand{\Id}{\mathbf{I}}

\newcommand{\ese}{\mathbf{S}}

\newcommand{\Te}{\mathbf{T}}
\newcommand{\Ge}{\mathbf{G}}
\newcommand{\J}{\mathbf{J}}
\newcommand{\M}{\mathbf{M}}
\newcommand{\Imonio}{\tilde{\mathbf{I}}}
\newcommand{\Ree}{\mathbf{R}}

\newcommand{\te}{\mathbf{T}}
\newcommand{\efe}{\mathcal{F}}
\newcommand{\ce}{\mathcal{C}}

\newcommand{\tr}{\textnormal{tr}}
\newcommand{\diag}{\textnormal{diag}}

\newtheorem{theorem}{Theorem}
\newtheorem{corollary}{Corollary}

\theoremstyle{definition}
\newtheorem{definition}{Definition}[section]
\newtheorem{propiedad}{Property}[section]

\theoremstyle{plain}

\newtheorem{lemma}{Lemma}[section]

\newcommand{\Ex}{\mathbb{E}}
\newcommand{\tildI}{\tilde{\mathcal{I}}(p)}
\newcommand{\de}{\mathbf{d}}

\usepackage{tikz,pgfplots}
\usetikzlibrary{positioning,shapes,babel}
\usetikzlibrary{external}
\tikzsetexternalprefix{Graficos/} 
\tikzsetfigurename{Fig} 

\pgfplotsset{tick label style={font=\scriptsize},
	legend style = {font=\small},
	xlabel style ={font = \footnotesize},
	ylabel style ={font = \footnotesize},
	yticklabel style={
	},		scaled ticks = false,
	legend style = {cells={anchor=west}},
	grid = both,
	every axis plot/.append style={line width=1pt}
}

\pgfplotsset{compat=newest}

\begin{document}
%
\title{On the Correlation between the Noise and \emph{a Priori} Error Vectors in Affine Projection Algorithms}
%
%
%

\author{Andr\'{e}s Altieri
\thanks{A. Altieri is with Universidad de Buenos Aires and CSC-CONICET, Buenos Aires, Argentina (email:  aaltieri@fi.uba.ar, andres.altieri@conicet.gov.ar).}
\thanks{This work was partially supported by the Grant UBACyT 20020190100330BA of Universidad de Buenos Aires and the International Cooperation Project CNRS-CONICET MOSIME.}}

\maketitle

\begin{abstract}
This paper analyzes the correlation matrix between the \emph{a priori} error and measurement noise vectors for affine projection algorithms (APA). This correlation stems from the dependence between the filter tap estimates and the noise samples, and has a strong influence on the mean square  behavior of the algorithm. We show that the correlation matrix is  upper triangular, and compute the diagonal elements in closed form, showing that they are independent of the input process statistics. Also, for white inputs we show  that the matrix is fully diagonal. These results are valid in the transient and steady states of the algorithm considering a possibly variable step-size. Our only assumption is that the filter order is large compared to the projection order of APA and we make no assumptions on the input signal except for stationarity. Using these results, we perform a steady-state analysis of the algorithm for small step size and  provide a new simple closed-form expression for mean-square error, which has comparable or better accuracy to many preexisting expressions, and is much simpler to compute. Finally, we also obtain expressions for the steady-state energy of the other components of the error vector.

%
%
%
%
%

\end{abstract}


%

\section{Introduction and Main Contributions}
%
%
%
%

\IEEEPARstart{A}{daptive} filters~\cite{LeoHernan} have played a major role in many signal processing
applications over the last few decades. The \emph{normalized least mean-squares} (NLMS) algorithm is a widely used alternative, mainly due to its good performance, ease of implementation and low
computational cost. As a downside, its rate of convergence is
sensibly reduced when colored inputs are used~\cite{Haykin}. In this context, the
\emph{affine projection algorithm} (APA)~\cite{UmedaOzeki_APA}, provides an increase in convergence speed with a modest increase in computational
complexity, while maintaining a robust behavior. 

The analysis of the convergence behavior of APA 
is very involved due to the nonlinear dynamics of the update equations, which introduces a strong correlation between the magnitudes involved. In particular, the filter tap estimates are correlated with previous inputs and noise samples via the filter update equations. This induces a statistical dependence between previous noise samples and the error vector obtained from the reference signal and the  estimates provided by APA. This correlation is key to analyzing the performance of APA but, for tractability reasons, some simplifying assumptions are usually invoked. In many cases, these assumptions include  simplifications on the input signal model~\cite{SB_2000,Ogunfunmi2011,Wan2011} or one or more independence assumptions~\cite{SB_2000,SS_2004,RRTB_2007,Song2011} between filter tap estimates, filter inputs, noise, or functions of these magnitudes.
For example, in \cite{SB_2000}, the authors perform a mean square (MS) analysis considering a simplified model for the input process. In this model~\cite{Slock}, the time-delayed input vector obtained from a stationary stochastic process is replaced by a sequence of independent random vectors which can take a finite number of orthogonal directions. In addition, they consider a strong hypothesis that past noise samples and filter coefficients are independent.  In \cite{SS_2004}, energy conservation arguments are used  to study the MS behavior of APA. Although there isn't a simplified model for the input process, the independence between past noise samples and filter coefficients is maintained and other independence hypotheses are added. Later works extended \cite{SB_2000,SS_2004} by  attempting to consider the correlation between the filter coefficients and past noises. For example, \cite{Ogunfunmi2011} extends the analysis in \cite{SB_2000} to consider the correlation between the noise and the error vector, but uses the same simplified model of the input signal. 
In \cite{Song2011} the authors extend \cite{SS_2004} considering the dependence on the noise and the filter tap estimates, but consider an independence assumptions between the tap estimates and a matrix obtained from input samples.
Finally in \cite{Wan2011}, they perform a MS analysis by developing recursive update equations for the correlation matrix of the filter tap error vector, considering its correlation with the noise. However, the analysis considers several simplifying assumptions, for example, simplifications for the input signal which are valid only asymptotically for white inputs. 

In this paper, we provide further insight on the MS behavior of APA by analyzing the correlation matrix between the additive noise vector $\ve_i = [v_i,...,v_{i-K+1}]^T$ and the \emph{a priori} error vector $\e_{a,i}$, which is the error between the signal estimate produced by APA and the reference signal without the noise~\cite{SS_2004} ($K$ is the APA projection order).
This correlation matrix has a strong impact on the transient and steady state MS behavior of the error vector. When the tap estimates are assumed to be independent of the noise, this matrix is assumed to be zero, which simplifies the analysis. In principle, this correlation matrix is very complex and depends on the input signal statistics. However, our analysis  shows that it is possible to compute its diagonal in closed form both in the transient and steady state of the algorithm. Furthermore, we show that this correlation depends only on the sequence of step-sizes and the noise variance, and is  independent of the input process. In addition, we show that for a white input process the matrix is diagonal, and hence, we fully characterize it. For our analysis we consider the standard APA with a variable step-size without regularization. In contrast to previous works, we do not take any assumptions on the input process and our only assumption is that the filter length $M$ is large compared to the projection order $K$ of APA.
Our simulations show that the obtained expressions are also valid in cases where this assumption does not hold.
The previous works mentioned before have analyzed many aspects of the MS behavior APA but have not analyzed the structure of this matrix. A related preceding work is \cite{RRTB_2007}, in which the authors find the trace of the correlation matrix for the case of a white input sequence for variably regularized and fixed step-size algorithm. Our work is an extension because we consider any input process, we compute the diagonal of the matrix instead of the trace and give information on off-diagonal elements. 
Using this analysis, we perform a new MS steady-state analysis of the behavior of the error vector for a small step-size. We provide a new simple closed-form expression for the steady-state mean-square error (MSE) of the algorithm which only depends on the noise variance, APA projection order and step-size. We show that, although very simple, this expression captures the behavior of the MSE for a small step-size. We compared this expression with other existing closed-form simple expressions from \cite{SB_2000,SS_2004,Ogunfunmi2011} obtained with and without the independence assumption which depend on the input signal statistics. We show that  our expression has a better or comparable accuracy and is very simple to compute. Finally, we provide a characterization of the steady-state energy of the other components of the error vector which, we believe, has not been done before.



The	 paper is organized as follows: in Section \ref{sec:defns} we introduce the model and notation used. In Section \ref{sec:Eeav} we present the analysis of the correlation between the noise and the \emph{a priori} error vector and in Section \ref{sec:error} we perform the steady state analysis. Finally, in Section \ref{sec:sims} we present simulation results, and in Section \ref{sec:closing} we present some closing remarks.

	\subsection{Notation}
We use boldface symbols for vectors (lower case) and matrices (upper case). $(\cdot)^T$ denotes transpose, $(\cdot)^H$ conjugate and transpose and $(\cdot)^*$ denotes conjugate. $\diag(\cdot)$ is the diagonal of a matrix, $\Ex$ is the expectation. For a matrix $\Ree$ we denote its $(i,j)$ element as $[\Ree]_{i,j}$. $\#$ denotes the cardinality of a set.


\section{System Model and APA recursion} \label{sec:defns}
We consider a reference signal $d_i$ which is generated using a linear  model: 
\begin{equation}
d_i = \x_i^H \w+v_i,
\end{equation}
where $\w \in \C^{M\times 1}$ is the finite impulse response of a causal system. The input signal vector $\x_i = [x_i,...,x_{i-M+1}]^T$ contains samples of a complex second-order stationary process with zero mean and $v_i$ is complex zero-mean white measurement noise of variance $\sigma_v^2$, independent of input.  The goal is to
obtain an estimate $\w_i$ of the unknown system vector $\w$ using the observations $\{d_n, \x_n\}_{n\leq i}$. This is achieved by feeding an adaptive filter the same input
signal as the unknown filter and adjusting its response using the error signal $e_i = d_i - \x_i^H \w_i$.
If we define the $M \times K$ data matrix $\X_i = \left[\x_i \ldots
\x_{i-K+1}\right]$, and the $K \times 1$ reference  data vector
$\de_i = [d_i  \ldots d_{i-K+1}]^T$, we can define an error vector 
	$\e_i = \de_i - \X_i^H \w_i$, 
which measures the error that  $\w_i$ will produce estimating the reference signal for the time range $i,...,i-K+1$. The standard APA recursion can be stated as~\cite{LeoHernan}: 
\begin{equation}
	\w_{i+1} = \w_i + \mu \X_i (\beta \Id_K + \X_i^H \X_i)^{-1} \e_i.  \label{APA}
\end{equation}
The parameter $\beta>0$ is a regularization which is included to improve the conditioning of the matrix  $\X_i^H \X_i$, $\mu$ is a step size parameter which controls the convergence of the algorithm, and $\Id_K$ is the $K\times K$
identity matrix.	
 Mean and mean square stability of (\ref{APA}) are guaranteed if $0 <
\mu \leq 2$; however, $0<\mu\leq1$ is preferred because it provides
a smaller steady-state error \cite{SB_2000}. In addition,
convergence speed and tracking are maximized when $\mu$ is close to
1, while the smallest steady-state error is achieved when it is
close to $0$. Using a large step-size when the estimation error is
large and decreasing it as the filter reaches steady-state, leads
to a fast initial convergence and small final error~\cite{PCB_2008, SS_2004_2}. For this reason, a time-varying step-size $\mu_i$ is generally used. For tractability, in this work we consider that $\beta=0$.

We also define a  misalignment vector $\wei = \w - \w_i$, which measures the difference between the estimate and the true system impulse response. With this we can formally write the \emph{a priori} error vector as $\e_{a,i} = \X_i^H \wei$. As mentioned before, this  indicates how $\w_i$ can estimate the vector $\de_i$ without
taking into account the additive noise component. Using the noise vector $\ve_i = [v_i \ldots v_{i-K+1}]^T$, we can
now  write the error vector as 
	$\e_i= \e_{a,i} + \ve_i.$

\section{Analysis of the correlation between the a priori error and noise} \label{sec:Eeav}
In this section, we analyze the  correlation matrix between the vectors  $\e_{a,i}$ and $\ve_i$. These vectors are correlated for $K>1$  because the noise samples $v_{i-1},...,v_{i-K+1}$ in $\ve_i$ appear in previous updates of  $\w_i$.  
We mention that $\Ex[\e_{a,i} \ve_i^H]_{1,1} = 0$ because the first component of the \emph{a priori} error vector is uncorrelated with the noise sample $v_i$. In particular for $K=1$, which corresponds to the NLMS algorithm, $\e_{a,i}$ and $\ve_i$ are uncorrelated scalar random variables.  
For $K>1$ the remaining elements of $\Ex[\e_{a,i} \ve_i^H]$ are unknown, except for the case of a white input signal for which the trace is known~\cite{RRTB_2007}. In what follows we analyze  this matrix and provide some insights on its structure. We start from the recursion of the misalignment vector  obtained  from (\ref{APA}):
\begin{equation}\label{eq:recwei}
	\wei = \weim  - \X_{i-1}  \ese_{i-1}  \e_{i-1}, 
\end{equation}
where $\ese_{i-1} = \mu_{i-1} (\X_{i-1} ^H\X_{i-1})^{-1}$.
By continuing the iteration into the past, the following recursion is obtained:
\begin{equation}  
	\we_i = \prod_{j=1}^K(\Id - \Ge_j) \we_{i-K} - \sum_{j=1}^{K} \left( \prod_{k=1}^{j} (\Id-\Ge_{k-1}) \right) \J_j \ve_{i-j}, \label{recursionerror1}
\end{equation}
where $\mathbf{J}_j = \X_{i-j} \ese_{i-j}$ and $\mathbf{G}_j = \J_j\X_{i-j}^H$, for $j>0$, and 
$\mathbf{G}_0 = 0$.  We can multiply  (\ref{recursionerror1})  by $\X_i^H$ and $\ve_i^H$ to obtain:
\begin{equation}
\Ex\left[\e_{a,i}\ve_i^H \right] = \Ex \left[\X_i^H\prod_{j=1}^K(\Id - \Ge_j) \we_{i-K}\ve_i^H \right] - \Ex\left[\X_i^H\sum_{j=1}^{K} \left( \prod_{k=1}^{j} (\Id-\Ge_{k-1}) \right) \J_j \ve_{i-j}\ve_i^H \right].
\end{equation}
The first term on the right side vanishes because the noise vector $\ve_{i}$ has zero mean and is independent of the rest. For $K=1$ the second expectation is also zero as mentioned before. For $K>1$ we then have:
\begin{equation}
	\Ex\left[\e_{a,i}\ve_i^H \right]  = - \sum_{j=1}^{K} \Ex\left[\X_i^H \left( \prod_{k=1}^{j} (\Id-\Ge_{k-1}) \right) \J_j \right] \Ex \left[\ve_{i-j}\ve_i^H \right]. \label{eq:Eav2}
\end{equation}
We then write
$\Ex\left[\ve_{i-j}\ve_i^H\right] = \sigma_v^2 \Imonio_{K,j}$
where $\Imonio_{K,m} \in \C^{K \times K}$, is a matrix such that ($m \in \Z$):
\begin{equation}
	\left[\Imonio_{K,m}\right]_{q,p} = \begin{cases} 1& \mbox{ if $p-q = m$}  \\ 0& \mbox{ otherwise.} \end{cases} \label{eq:Imonio}
\end{equation}
We notice that $\Imonio_{K,K} = \mathbf{0}$, which means that the $K$-th term in the second expectation in (\ref{eq:Eav2}) vanishes. Then we have:
\begin{equation}
\Ex\left[\e_{a,i}\ve_i^H \right] = - \sigma_v^2E\left[\X_i^H \J_1\right] \Imonio_{2,1}  \ \ \ (K = 2),
\end{equation}
and for $K>2$:
\begin{equation}
	\Ex\left[\e_{a,i}\ve_i^H \right] =- \sigma_v^2 \left(\Ex\left[\X_i^H \J_1\right]  \Imonio_{K,1} +\sum_{j=2}^{K-1} \Ex\left[\X_i^H\left( \prod_{k=1}^{j-1} (\Id-\Ge_{k}) \right) \J_j\right] \Imonio_{K,j}\right). \label{eavwhite2}
\end{equation}
The expression for $K>2$ is involved because of matrix product  $\prod_k (\Id-\Ge_{k})$ which is non-commutative. 
We now introduce an approximation for products of matrices of the form $\X_i^H \X_{i-j}$ which appear in (\ref{eavwhite2}). When the length $M$ of the filter is large compared to the projection order $K$, it is reasonable to take the approximation $\X_i^H \X_i \approx M \Ree_{x}$, where $\Ree_{x} \in \mathbb{C}^{K \times K}$ is the correlation matrix of the input process~\cite{RRTB_2007}. In the same way it is possible to approximate
\begin{equation}
\X_i^H \X_{i-j} \approx M \Ree_{x, -j} \label{eq:approxR}
\end{equation}
where $\Ree_{x,m} =  \Ex\left[ [x_i,...,x_{i-K+1}]^H [x_{i+m},...,x_{i+m-K+1}]\right]\in \mathbb{C}^{K \times K}$ is a time shifted correlation matrix.
Using (\ref{eavwhite2}) as starting point and considering approximation (\ref{eq:approxR}) we introduce our main result for this section:
\begin{theorem} \label{teo:main}
When $M \gg K$, assuming that (\ref{eq:approxR}) is a valid approximation,  for any input process $x$,  $\Ex\left[\e_{a,i} \ve_i^H\right]$ is an upper triangular matrix, and its diagonal elements are:
\begin{equation}
	\Ex\left[\e_{a,i} \ve_i^H\right]_{q,q} = -\sigma_v^2 \left(\mu_{i-1} + \sum_{j=2}^{q-1} \mu_{i-j}  \prod_{k=1}^{j-1} (1-\mu_{i-k}) \right) \ \ \ (2 \leq q \leq K). \label{aprox13}
\end{equation}
Furthermore, if the input process is white then $\Ex\left[\e_{a,i} \ve_i^H\right]$ is a diagonal matrix, that is, its only non-zero elements are given by (\ref{aprox13}).
\end{theorem}
\begin{proof}
	It is presented in Appendix A.
\end{proof}
From Theorem \ref{teo:main} we outline the following conclusions:
\begin{itemize}
	\item The elements below the main diagonal of $\Ex\left[\e_{a,i} \ve_i^H\right]$ are zero.
	\item The diagonal of $\Ex\left[\e_{a,i} \ve_i^H\right]$ can be computed in closed form it depends only on the noise variance $\sigma_v^2$ and the sequence of step sizes $\{\mu_i\}$, that is, it does not depend on the statistics of the input process $x$. This is true in the transient and steady-state of the algorithm.
	\item In general the elements above the main diagonal will depend on the statistics of the input process $x$ except for the white input case, when they are zero.
\end{itemize} 
Considering a fixed step-size we have the following corollary which will be used in the following section:
\begin{corollary} Considering a constant step-size $\mu_{i} \equiv \mu$ in (\ref{aprox13}) we find that:
\begin{equation}
\Ex\left[\e_{a,i} \ve_i^H\right]_{q,q} = -\sigma_v^2 \left[1-(1-\mu)^{q-1}\right] \ \ (1 \leq q \leq K). \label{eq:diagEav}
\end{equation}
Notice that, in contrast to (\ref{aprox13}), this is also valid for $q=1$.
Summing the elements of the diagonal we find that:
\begin{equation}
	\Ex\left[\e_{a,i}^H \ve_i\right] = -\sigma_v^2 \left[K - \frac{1- \left(1-\mu\right)^K}{\mu}\right]. \label{eq:EeaHv}
\end{equation}
\end{corollary}
In the following section we use these results to derive a new steady-state analysis of the algorithm.

\section{A steady-state error analysis without the independence assumption} \label{sec:error}

In this section  we perform a small step-size analysis of the steady state behavior of the algorithm. We derive a new formula for the steady state value of the MSE and we also find the steady-state energy of the other components of the error vector $\e_i$. This is done under the conditions of Theorem \ref{teo:main} ($M\gg K$) and under a weak approximation on the energy of the error vector $\e_{a,i}$ in the steady state.

%
We start from the recursion of the error vector (\ref{eq:recwei}) and multiply by $\X_{i}^H$ to obtain:
\begin{equation}
	\X_i^H \weii = \X_i^H \wei  -  \mu_i \e_i.
\end{equation}
By taking the squared norm and the expectation we find that:
\begin{equation}
	\Ex\left[||\X_i^H \weii||^2\right] = \Ex\left[||\X_i^H \wei||^2\right] + \mu_i^2 \Ex\left[ ||\e_i||^2\right] -  \mu_i \Ex\left[\e_i^H \e_{a,i} + \e_{a,i}^H \e_i\right]. \label{eq:SSnormE}
\end{equation}
As $i\to \infty$, in the steady state condition, the step-size is assumed to converge to a small steady-state value $\mu \ll 1$. Alternatively, we can assume that the step-size is fixed and small for all time-instants ($\mu_i \equiv \mu \ll 1$). When the filter and the step-size have converged, the filter updates will be very small such that the following approximation is reasonable:
\begin{equation}
	\Ex\left[||\X_i^H \weii||^2\right] \approx \Ex\left[||\X_i^H \wei||^2\right] \ \ \ \text{(as $i \to \infty$).} \label{eq:ap}
\end{equation}
This means that, on average, there will be almost no reduction in the error after an update in the steady-state small-step size condition. Using this approximation in (\ref{eq:SSnormE}) we may obtain the following steady-state equation:
\begin{equation}
	\mu\Ex\left[||\e||^2\right]_\infty = \Ex \left[ \e^H  \e_{a}\right]_\infty+ \Ex \left[\e_{a}^H \e\right]_\infty, \label{eq:newap}
\end{equation}
where the subscript $(\cdot)_\infty$ denotes the steady-state values of the expectations as $i\to \infty$. This equation is valid in principle for any input process, provided that $\mu \ll 1$. 
We would like to compare (\ref{eq:newap}) to the general steady-state equations that were presented in \cite{SS_2004}:
\begin{equation}
	\mu \Ex \left[\e_i^H (\X_i^H \X_i)^{-1}\e_i \right] = \Ex \left[\e_{a,i}^H (\X_i^H \X_i)^{-1}\e_i+ \e_{i}^H (\X_i^H \X_i)^{-1}\e_{a,i}\right] \ \ (i\to \infty). \label{eq:Shin2}
\end{equation}
This equation is valid without any approximation whatsoever. It is interesting to notice that if in (\ref{eq:Shin2}) we consider a white input and for $M\gg K$ we approximate $\X_i^H \X_i \approx M \sigma_x^2 \Id$ we obtain (\ref{eq:newap}). This means that (\ref{eq:newap}) will be exact for a white input process for any value of $\mu$, but our analysis shows   it will  be a good approximation for any other input process when $\mu \ll 1$, that is, all input processes will behave as white processes in the steady-state, provided that $\mu$ is small enough.

We can now use the results of Theorem \ref{teo:main} to characterize the steady-state value of the error vector without using an independence assumption. Using that $\e_i = \e_{a,i}+\ve_i$ in (\ref{eq:newap}) we can derive an expression for the energy of the error vector:
\begin{equation}
	\Ex\left[||\e||^2\right]_\infty = \frac{2}{2-\mu}\left(\Ex\left[||\ve||^2\right]_\infty+ \Ex\left[\e_{a}^H\ve\right]_\infty + \Ex\left[\ve^H \e_{a}\right]_\infty\right). \label{eq:normaee}
\end{equation}
Using (\ref{eq:EeaHv}) in the (\ref{eq:normaee}), the steady-state value of the energy of the error vector  evaluates to:
\begin{equation}
	\Ex\left[||\e||^2\right]_{\infty} = \frac{2\sigma_v^2}{\mu(2-\mu)} \left(1- (1-\mu)^K \right). \label{eq:normaee2}
\end{equation}
Also, for the energy of the a priori error vector we  can find a closed form equation:
\begin{equation}
	\Ex\left[||\e_{a}||^2\right]_\infty = \frac{\mu}{2-\mu} \Ex\left[||\ve||^2\right]_\infty +  \frac{\mu-1}{2-\mu} \left( \Ex\left[\e_{a}^H\ve\right]_\infty + \Ex\left[\ve^H \e_{a}\right]_\infty\right), \label{eq:normaea}
\end{equation}
which can be evaluated using (\ref{eq:EeaHv}). In order to find the steady-state energy of the components of the error vector we introduce a simple assumption: 
\begin{itemize}
	\item A1) In the steady-state, with a small step-size, the components of the a priori error vector have the same energy:
	\begin{equation}
		\Ex\left[\e_{a}\e_{a}^H\right]_{q,q,\infty} \approx \frac{\Ex\left[||\e_{a}||^2\right]_\infty}{K}, \ \ (1 \leq q \leq K). \label{eq:appea2}
	\end{equation} 
\end{itemize}
Under this approximation we can write the excess mean square error (EMSE) as:
\begin{equation}
	\Ex\left[||\e_{a}||^2\right]_\infty \approx K \Ex\left[|e_{a}|^2\right]_\infty \label{eq:appea}
\end{equation}
where $e_{a,i} =d_i - \x_i^H \wei$ is the first component of $\e_{a,i}$.
Now replacing (\ref{eq:appea}) in (\ref{eq:normaea}) and noting that $\Ex\left[\ve_i^H \e_{a,i}\right]$ is a real magnitude we can write the excess mean square error (EMSE) in the steady-state as:
\begin{equation}
	\Ex\left[|e_{a}|^2\right]_\infty = \frac{\mu\sigma_v^2}{2-\mu}  +  \frac{2 (\mu-1)}{K(2-\mu)}   \Ex\left[\e_{a}^H\ve\right]_\infty. \label{eq:normaea2}
\end{equation}
The first term in (\ref{eq:normaea2}) is the formula for the EMSE for small $\mu$ presented in \cite{SS_2004} using the independence assumption and others. This expression does not consider the important dependence of the EMSE with $K$. 
The second term is a new additional term which depends on $K$ and under (\ref{eq:appea}) allows us to correct the EMSE using the correlation between the noise and the \emph{a priori} error vector. Finally, using (\ref{eq:EeaHv}) we can obtain a new and simple expression for the steady state EMSE as:
\begin{equation}
	\Ex\left[|e_{a}|^2\right]_\infty = \sigma_v^2 \left[1 - \frac{2 (1-\mu) \left(1 - (1-\mu)^{K}\right)}{K \mu(2-\mu)} \right].  \label{eq:Eea2}
\end{equation}
Finally, the energy of the $1 \leq q \leq K$ components of the error vector is given by:
\begin{equation}
	\Ex\left[\e_i\e_i^H\right]_{q,q} = \Ex\left[\e_{a,i}\e_{a,i}^H\right]_{q,q} +  \Ex\left[\e_{a,i} \ve_i^H\right]_{q,q}  
	+ \Ex\left[ \ve_i\e_{a,i}	^H\right]_{q,q} + \sigma_v^2.
\end{equation}
Under approximation A1 all the components of the \emph{a priori} error vector have energy given by (\ref{eq:Eea2}). Then, using (\ref{eq:Eea2}) and (\ref{eq:diagEav})  we can obtain the energy of each component of the error vector as ($1 \leq q \leq K$):
\begin{equation}
	\Ex\left[\e\e^H\right]_{q,q,\infty} =  \max\left\{0, 2 \sigma_v^2 \left[ (1-\mu)^{q-1} - \frac{(1-\mu)\left(1-(1-\mu)^K\right)}{K(2-\mu)\mu}\right]\right\}. \label{eq:diagEeeH}
\end{equation}
The $\max$ is included because it can be seen that the expression will become negative as $\mu \to  1$ for $q>1$. By setting $q=1$ in this expression we obtain the mean square error, which may be also obtained by adding $\sigma_v^2$ to (\ref{eq:Eea2}). 

\section{Numerical Results} \label{sec:sims}
In this section we perform simulations to explore the validity of the proposed expressions.  We consider that the impulse response of the true system $\mathbf{w}$ corresponds to a real measured acoustic impulse response which can be seen in Fig. \ref{fig:impulse_resp}.
We consider responses of length $M=512$, where the hypothesis $M\gg K$ is valid, and a shorter filter $M=32$ to show that the derived expressions are also valid in this case. For $M=512$ the first 512 samples are used, while for $M=32$ we use samples 106 to 137 of the impulse response.
We consider that the input signal $x_i$ is either a white Gaussian process, a first order Gaussian autoregressive process (AR1) with its pole close to 1, or a Gaussian autoregressive moving average process ARMA(2,2), obtained by filtering white Gaussian noise through a causal system with transfer function:
\begin{equation}
	H(z) = \frac{1-z^{-2}}{1 - 1.70223z^{-1} +0.71902z^{-2}}. \label{eq:arma22}
\end{equation}
This process is highly correlated, much more than the AR1 process, its
correlation matrix can have a conditioning number
larger than $10^5$ \cite{RRTB_2007}.
The additive noise $v_i$ is white Gaussian noise, with variance $\sigma_v^2$  adjusted to achieve a prescribed signal-to-noise ratio:
\begin{equation}
	\text{SNR} = 10 \log_{10}\frac{\Ex[|\x_i^H \w|^2]}{\sigma_v^2}.
\end{equation}
For all our simulations we consider an $\text{SNR}$ of 30dB. We denote by $J$ the number of averaged runs to obtain each plot. Finally, the APA algorithm is initialized with a zero initial condition.
\begin{figure}[t!]
	\centering
	\includegraphics[width=0.5\linewidth, trim=.5cm .6cm  0.4cm .5	cm,clip]{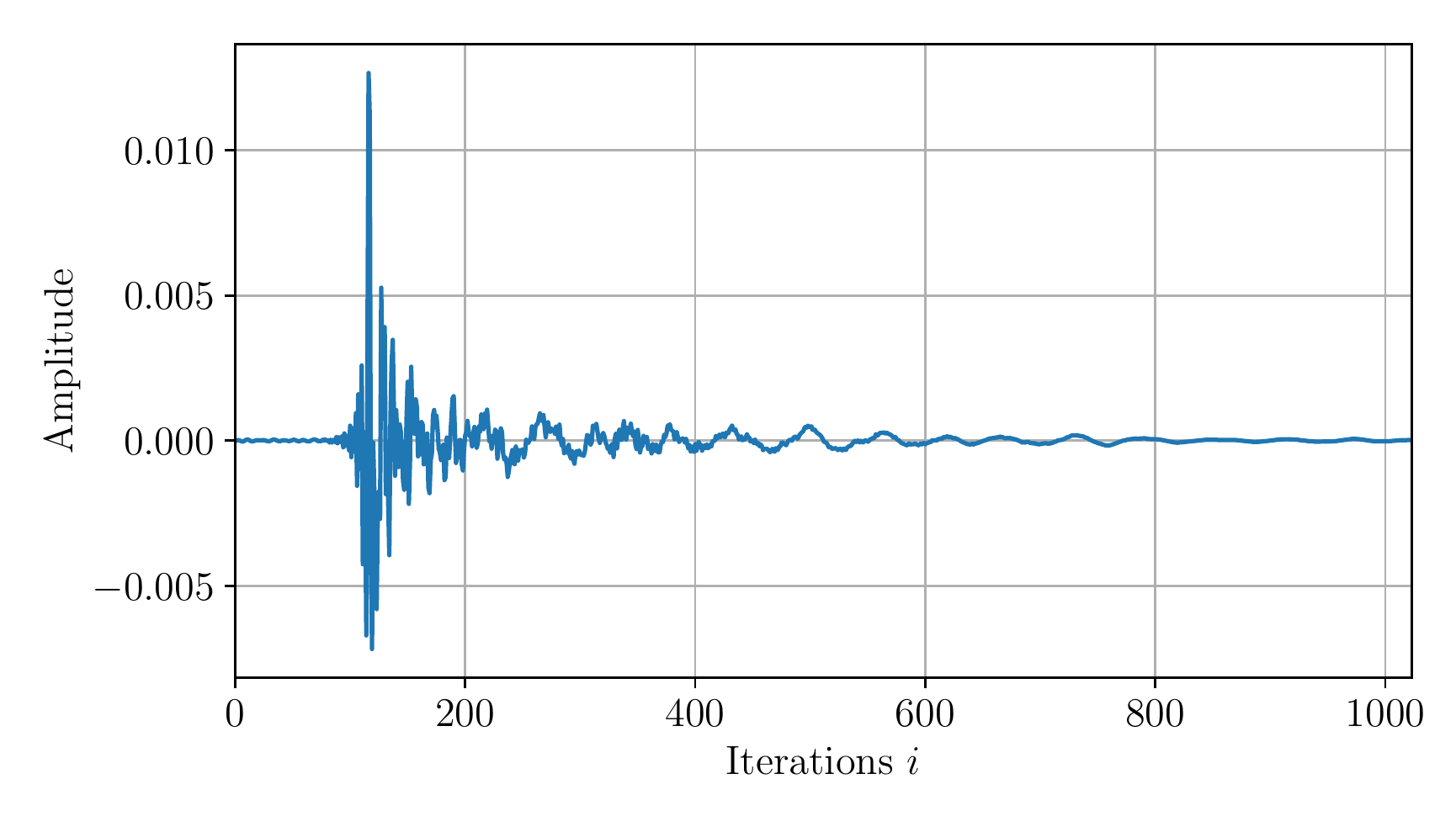}
	\caption{Measured acoustic impulse response to be used in the simulations.}
	\label{fig:impulse_resp}
\end{figure}

\subsection{Correlation between the a priori error and the noise}
We first explore the validity of the results from Theorem \ref{teo:main} about the  $\Ex\left[\e_{a,i} \ve_i^H\right]$. These results are valid for the transient and steady-state conditions of the algorithm, and for any time varying step-size sequence. In general, variable step-size algorithms start from a large step-size which becomes smaller during convergence, so we choose a  sequence with this behavior. Any deterministic sequence can be used; in particular we consider the following expression:
\begin{equation}
	\mu_i = \left[1+(C\times i)^2\right]^{-1}, \label{eq:museq}
\end{equation}
where in our case we choose $C=7\times 10^{-4}$. The time evolution of this sequence can be seen in Fig. \ref{fig:mu_seq}.
\begin{figure}
	\centering
	\includegraphics[width=.5\linewidth, trim=.5cm .8cm  0.6cm .7cm,clip]{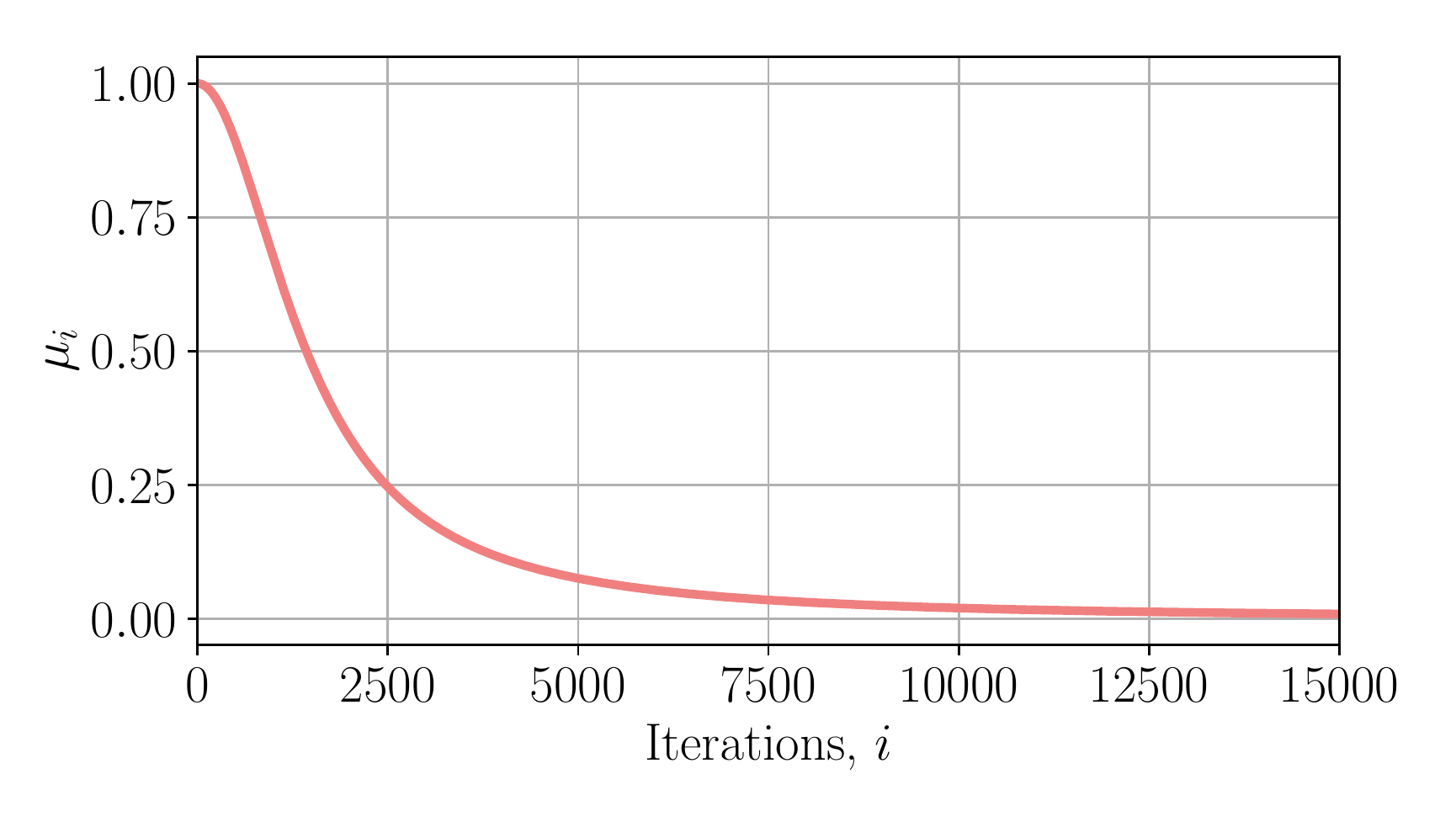}
	\caption{Sequence of step size $\mu$ to use in simulations, given by (\ref{eq:museq}) with $C=7\times 10^{-4}$.}
	\label{fig:mu_seq}
\end{figure}

In Fig. \ref{fig:diagEav} we compare the results of Theorem \ref{teo:main} for different input signals with $M=512$ and $K \leq 8$, where the condition  $M \gg K$ of Theorem \ref{teo:main} holds.  
In particular, in Fig. \ref{fig:eavk4m512white} we consider $K=4$ and a white input sequence. In this case, the matrix was shown to be diagonal, with diagonal values given by (\ref{aprox13}). We observe that the proposed expressions are very accurate. 
\begin{figure}
	\centering
{\includegraphics[width=.5\linewidth, trim=.7cm .8cm  0.1cm .6cm,clip]{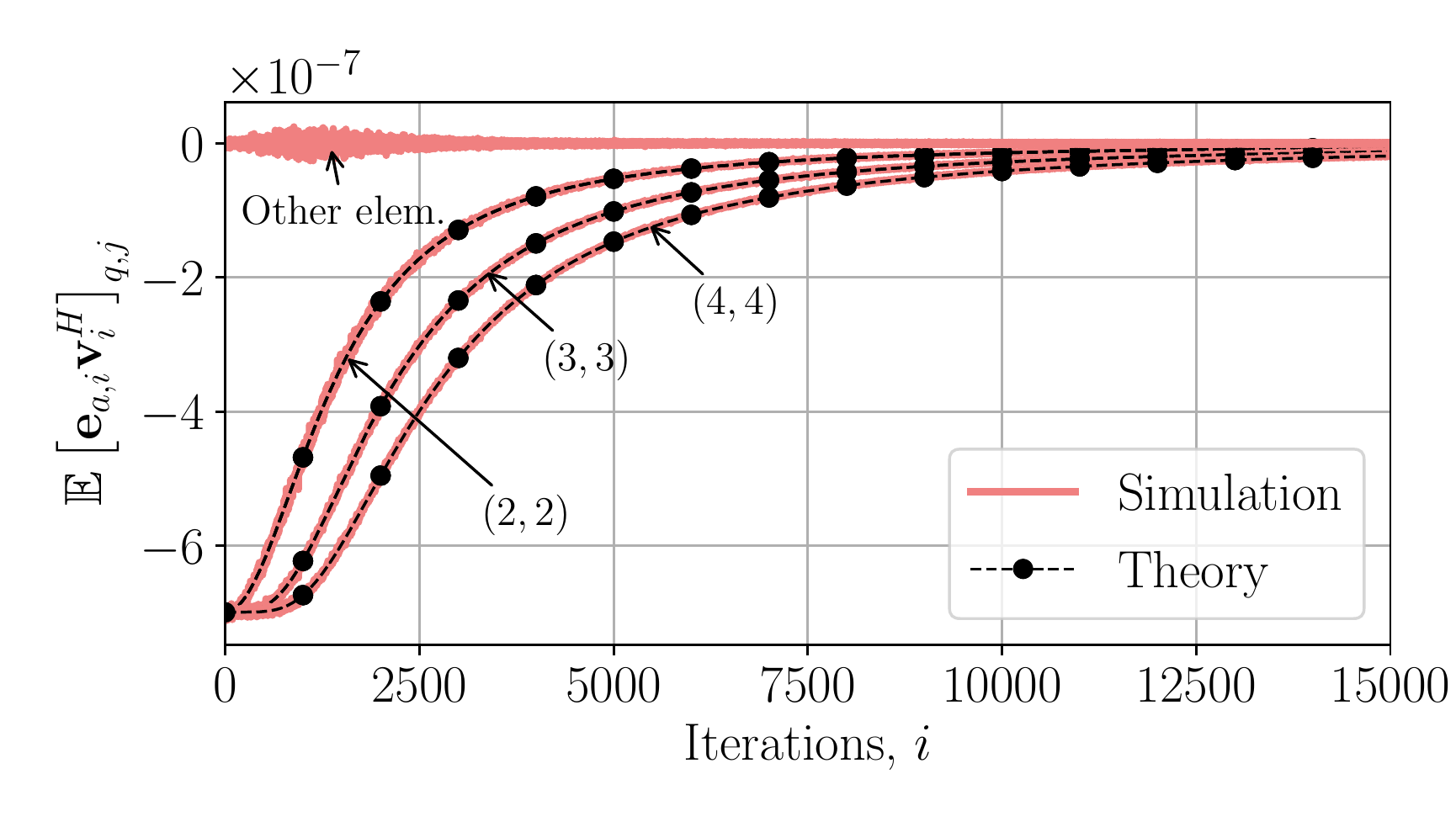}
	\subcaption{White input, $K=4$.} 	\label{fig:eavk4m512white}}
{\includegraphics[width=0.5\linewidth, trim=.7cm .8cm  0.1cm .6cm,clip]{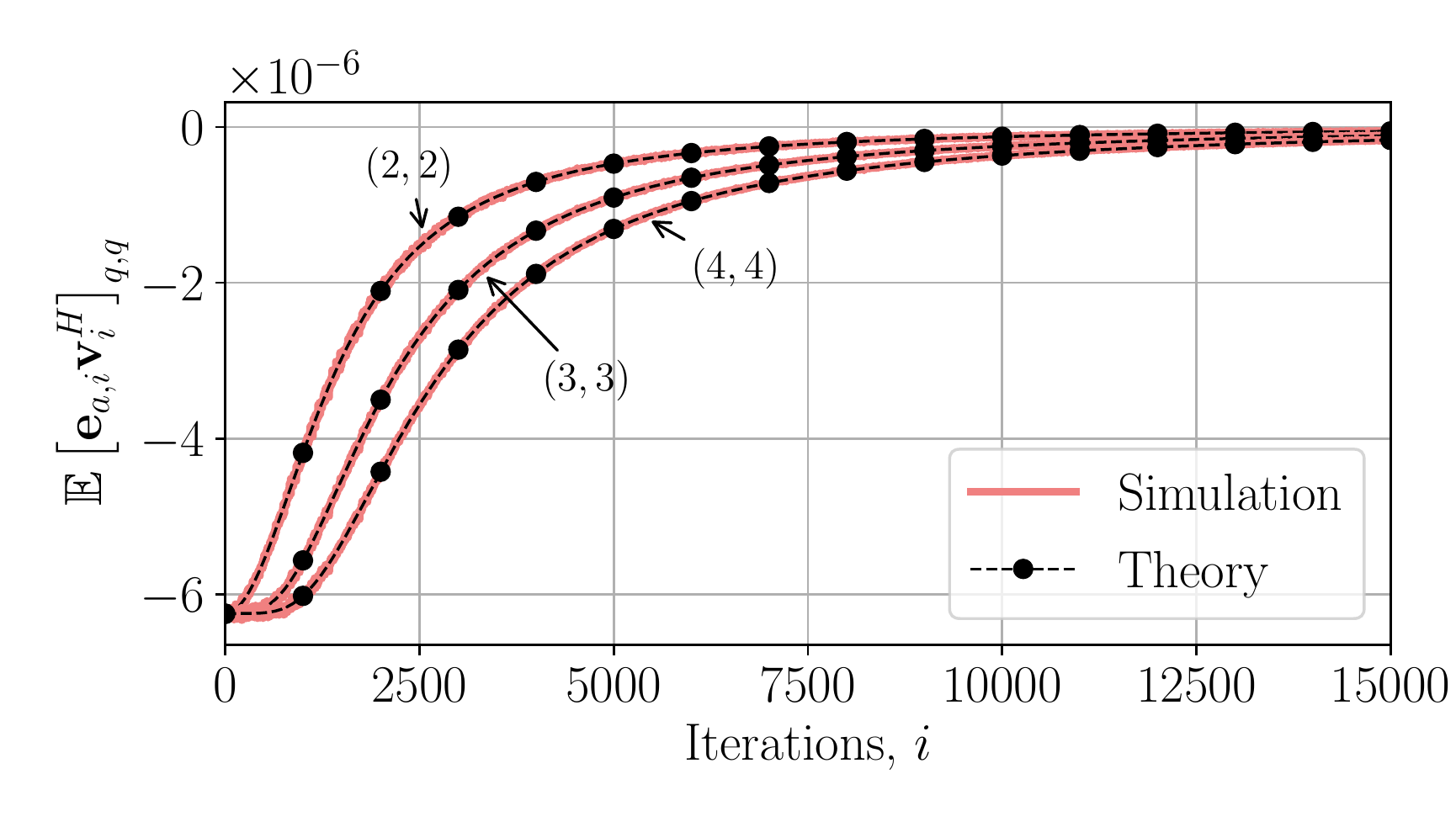}
	\subcaption{ARMA(2,2) input, $K=4$.} \label{fig:diageavk4m512arma22}}
{\includegraphics[width=.5\linewidth, trim=.7cm .8cm  0.1cm .6cm,clip]{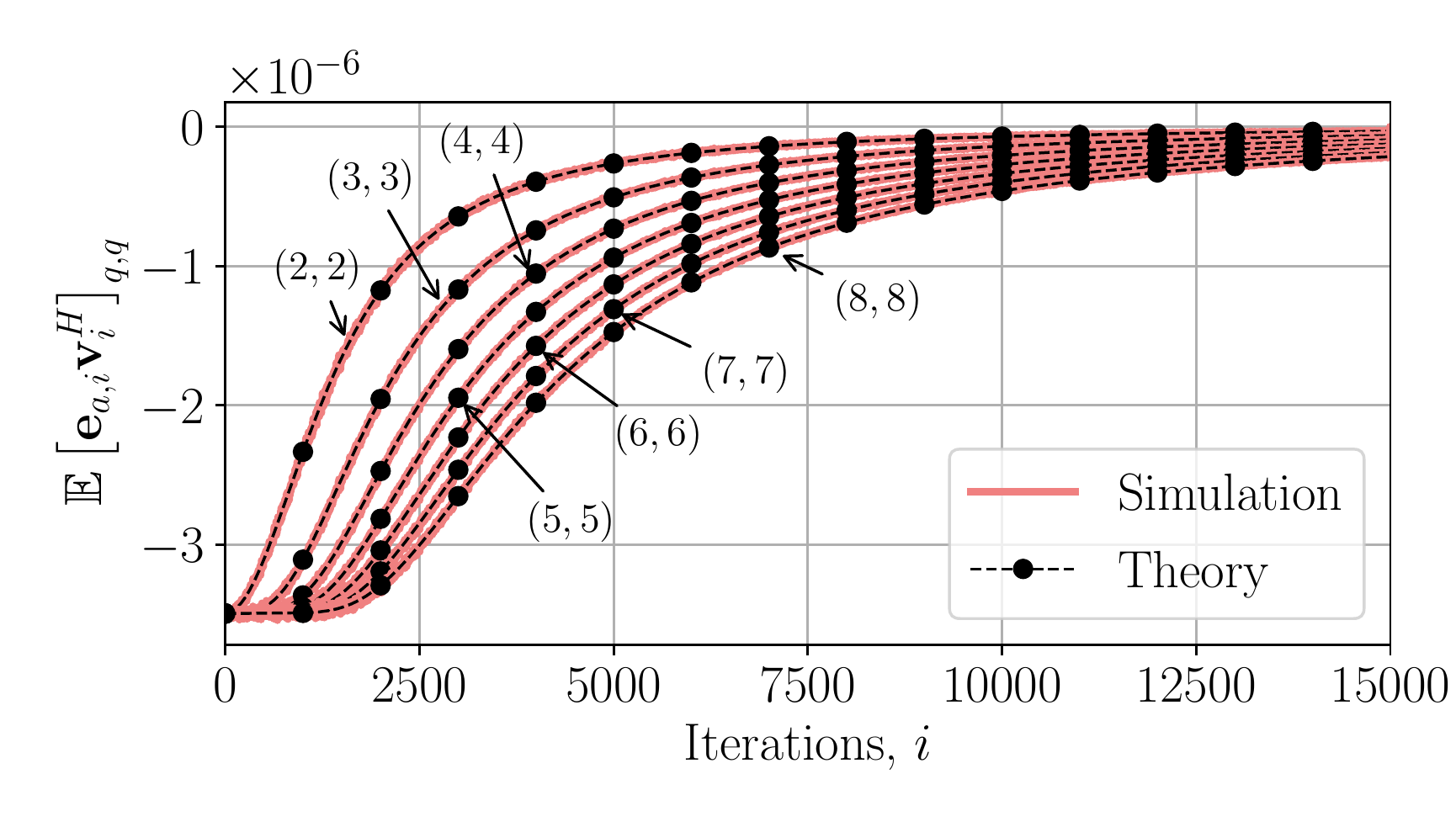}
	\subcaption{AR1(0.95) input, $K=8$.} \label{fig:diageavk8m512ar1}}
	\caption{Comparison of Monte Carlo simulations of $\Ex\left[\e_{a,i}\ve_i^H \right]$ compared to (\ref{aprox13}) and other results of Theorem \ref{teo:main} for different input processes. $M=512$. $\text{SNR}=30$dB. $J=10^5$.} \label{fig:diagEav}
\end{figure}
In Figs. \ref{fig:diageavk4m512arma22} and \ref{fig:diageavk8m512ar1} we consider $K=4$ with ARMA(2,2) input, and $K=8$ with AR1(0.95) input, respectively. For this case, $\diag\Ex\left[\e_{a,i}\ve_i^H \right]$ is known to be an upper triangular matrix with main diagonal values given by (\ref{aprox13}). In both cases we see that the simulations are in excellent agreement with the theoretical expression, even for the highly correlated ARMA(2,2) process. We do not plot the first element of the diagonal since we know that it is zero.

In Fig. \ref{fig:diagEav2} we consider cases in which $M \gg K$ does not hold. For conciseness, we focus on colored inputs. In  particular we consider $M=32$ with $K=4$ and AR1(0.95) input, or $K=8$ and ARMA(2,2) input. We can see that (\ref{aprox13}) is also very accurate in this case, and there seems to be no degradation in the accuracy from the previous case. As a general conclusion, we can see that the expressions put forth in Theorem \ref{teo:main} are  accurate.
\begin{figure}[t!]
	\centering
	{	\includegraphics[width=.5\linewidth, trim=1cm .8cm  0.1cm .6cm,clip]{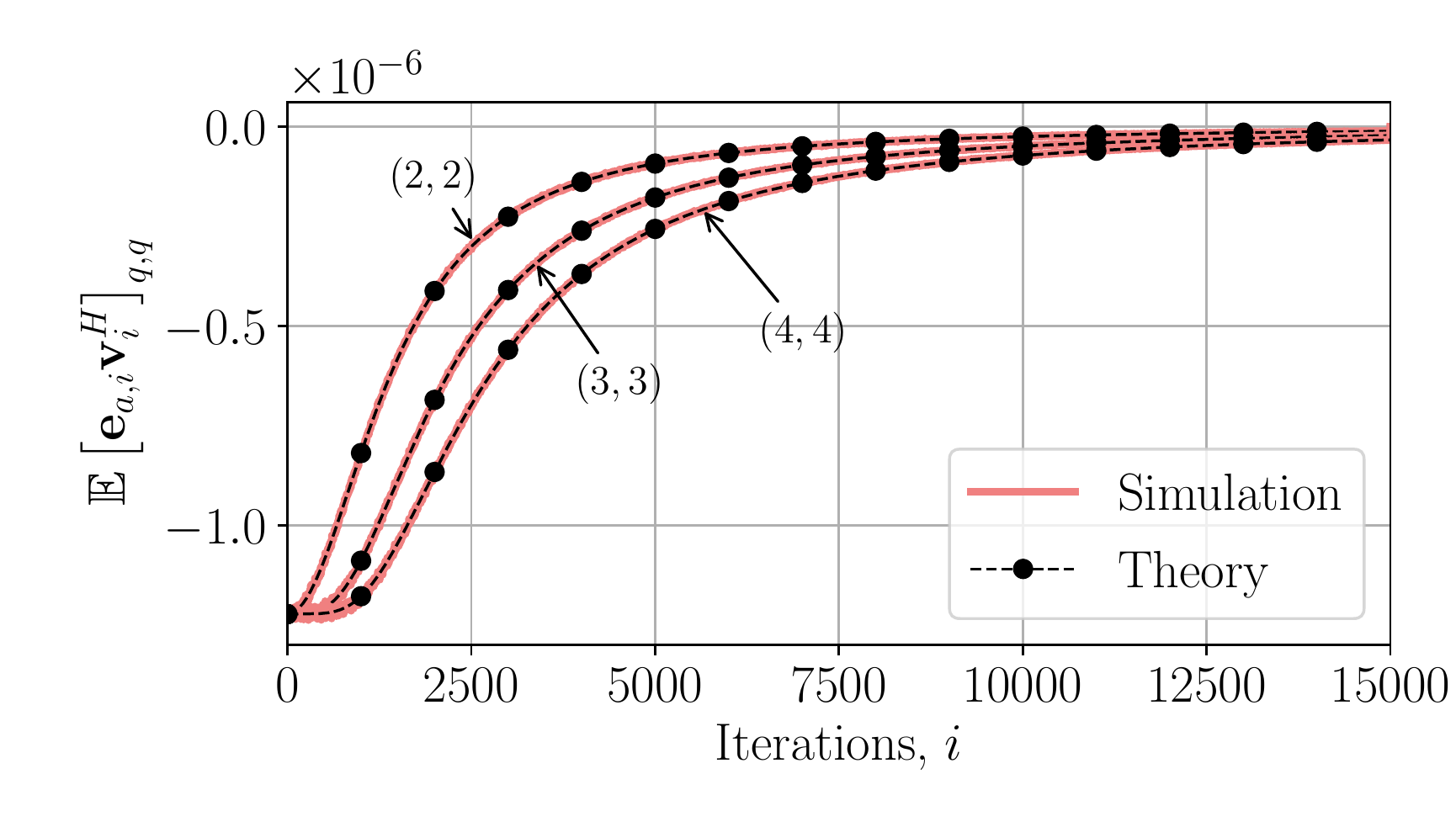}
		\subcaption{AR1 (0.95), $K=4$.}}
	{
	\includegraphics[width=.5\linewidth, trim=.7cm .8cm  0.1cm .6cm,clip]{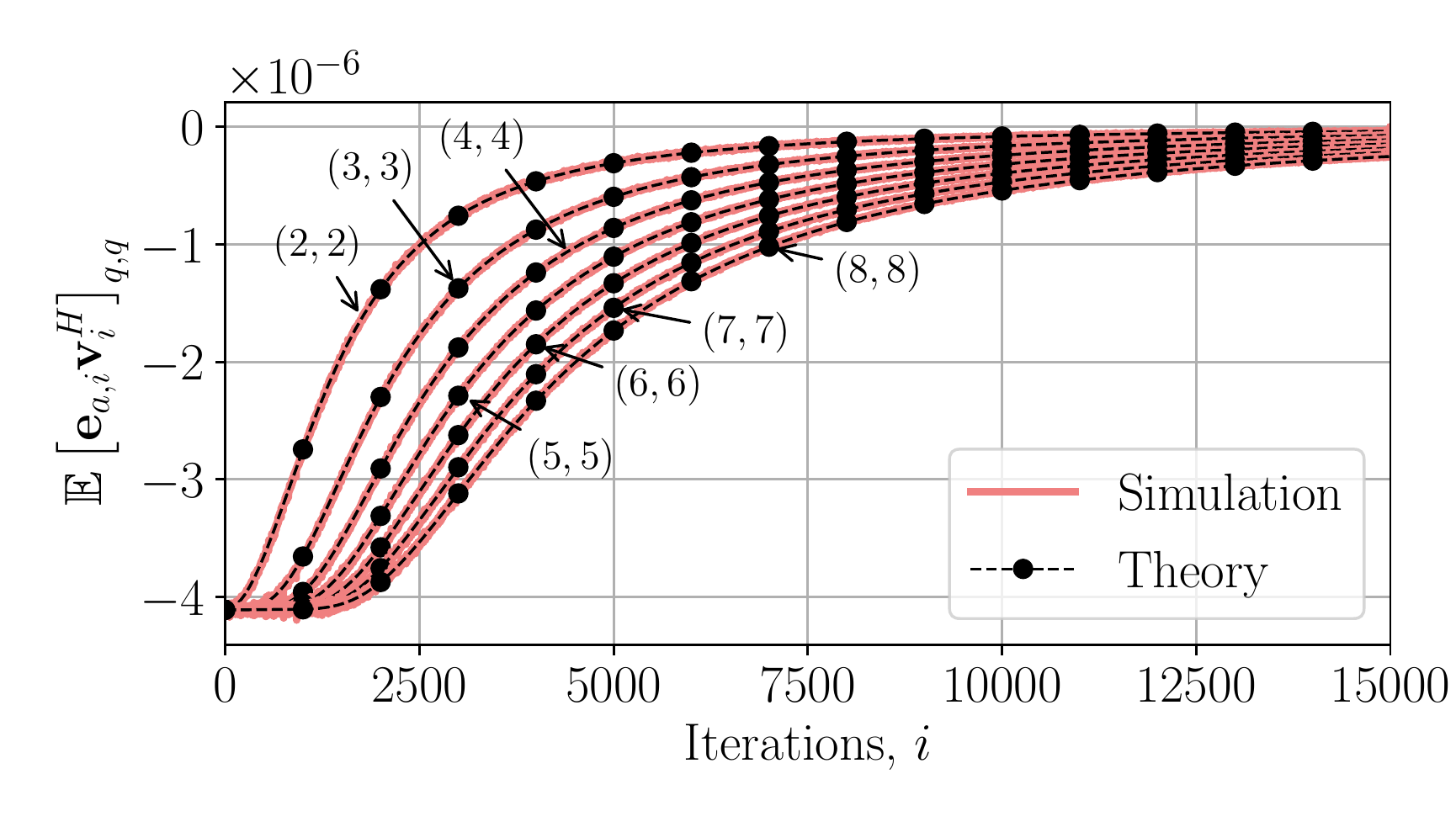}
	\subcaption{ARMA(2,2) input, $K=8$.} }
\caption{Comparison of Monte Carlo simulations of $\Ex\left[\e_{a,i}\ve_i^H \right]$ compared to (\ref{aprox13}) and other results of Theorem \ref{teo:main} for different input processes. $M=32$. $\text{SNR}=30$dB. $J=10^5$.}  \label{fig:diagEav2}
\end{figure}

%
\vspace{-3mm}

\subsection{Steady-state error analysis}

In this section we explore the validity of the steady-state expressions of Section \ref{sec:error}. First, we explore the validity of the steady-state condition (\ref{eq:newap}). To do this, we simulate (\ref{eq:normaee2}), which was obtained directly from (\ref{eq:newap})  and is more interesting.
We like to mention that (\ref{eq:normaee2})  was previously developed for white input signals in \cite{RRTB_2007} for any value of $\mu$ so we are only interested in showing that it is also a good approximation  for any input process as long as $\mu \ll 1$.  In Fig. \ref{fig:Enorme} we compare (\ref{eq:normaee}) for the case of the AR1(0.95) and the ARMA(2,2) processes, for different values of $K$ and $M$. Each curve was obtained by averaging $J=400$ runs of the error to convergence and averaging the last 2000 ($M=512$) or 200 samples ($M=32,64$) of the steady-state value. We can see that when $M$ larger than $K$, the approximation is very good for small $\mu$, even for very correlated signals such as the ARMA(2,2) process. This can be seen when $M=512$ and for all the proposed $K$, and when $M=32$ and $K=2,4$. The case when the greater divergence occurs is when $M=32$ and $K=8$, when $M$ is closer to $K$. For this we have included also the case $M=64$ to show that the formula becomes quickly more accurate when increasing $M$. Overall we can see that (\ref{eq:ap}) is a good approximation when $\mu \ll 1$ , even for very correlated signals, as long as $M$ is reasonably larger than $K$ ($M \geq 4K$ could be a good rule of thumb).  
\begin{figure}
	\centering
	{
	\includegraphics[width=.5\linewidth, trim=.5cm 0.5cm  0.2cm 0.5cm,clip]{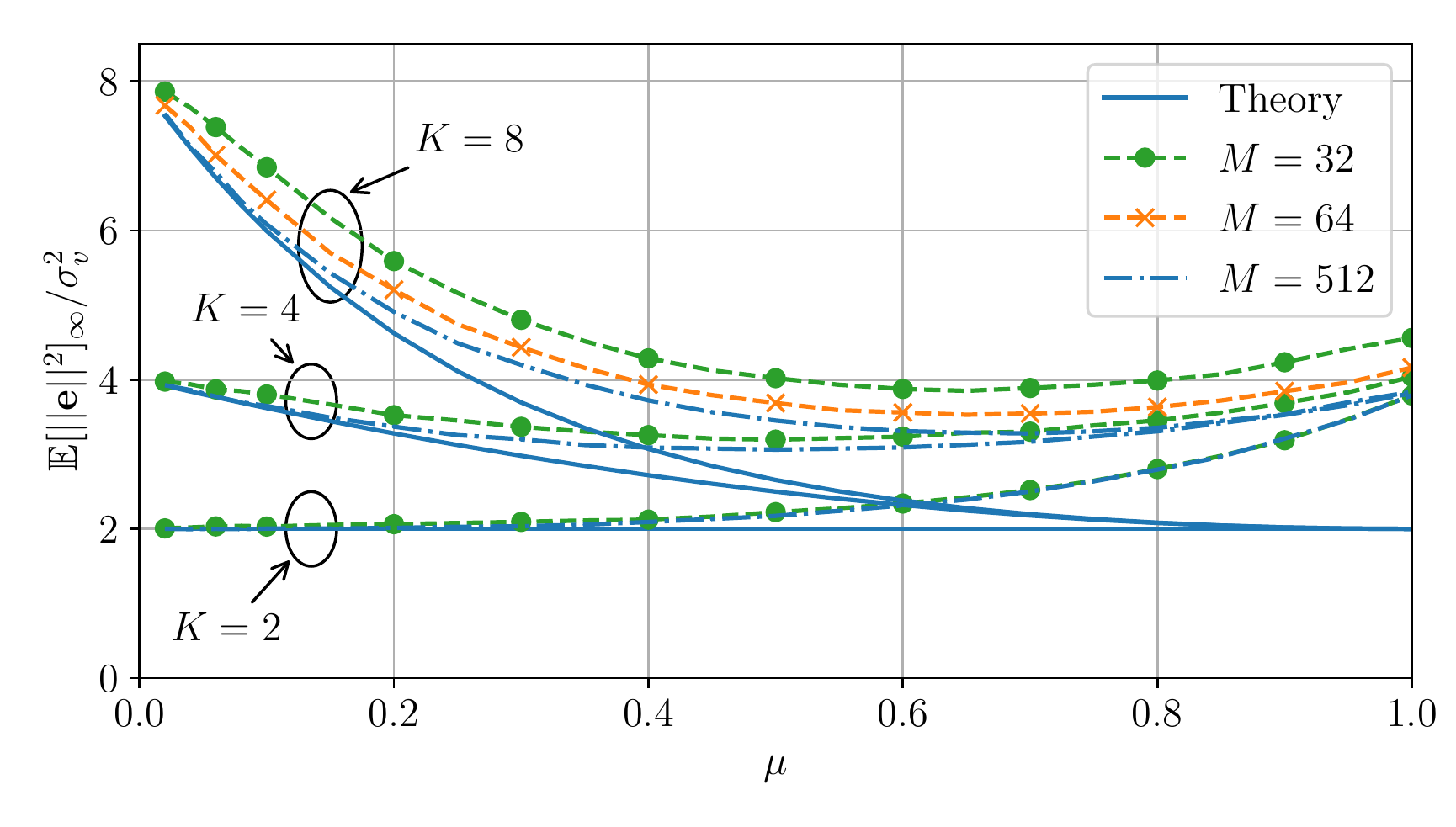}
	\subcaption{AR1(0.95)}}
	{
	\includegraphics[width=.5\linewidth, trim=.5cm 0.5cm  0.2cm 0.5cm,clip]{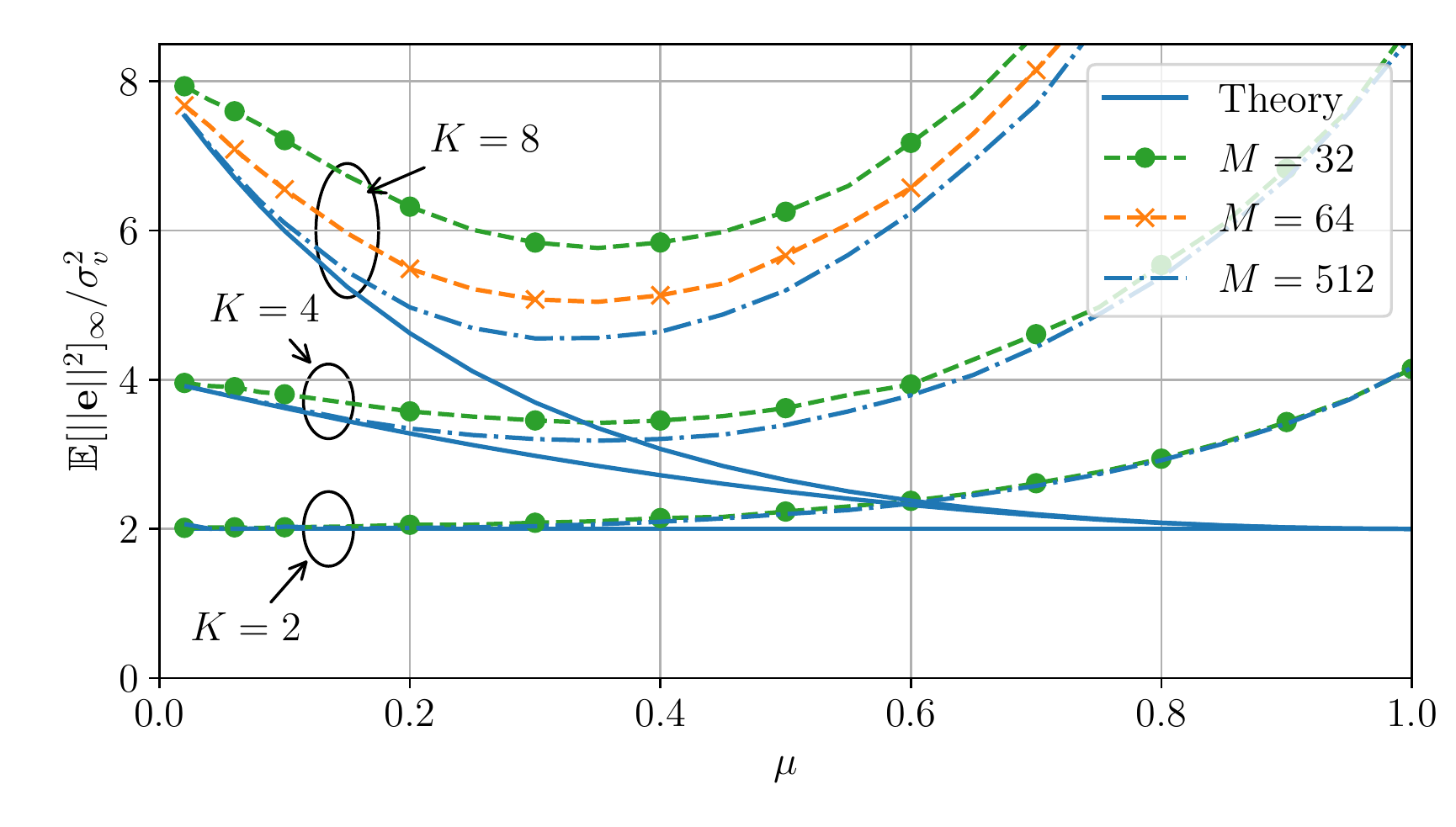}
	\subcaption{ARMA(2,2)}}

	\caption{Steady-state of the square norm of the error vector for different colored inputs. The theoretical expression is given by (\ref{eq:normaee2}). $\text{SNR}=30$dB. $J=400$.}
	\label{fig:Enorme}
\end{figure}

We now test the expression of the MSE obtained by taking $q=1$ in (\ref{eq:diagEeeH}) comparing it with simple expressions that have appeared in other works. We consider an expression valid for small $\mu$~\cite{SS_2004}:
\begin{equation}
	\text{MSE} =   \frac{2\sigma_v^2}{2-\mu}. \label{eq:MSE}
\end{equation} 
which is obtained using the independence assumption and others. 
We also consider an expression which accounts for the correlation between the noise and the error, is valid for any $\mu$, and depends on the second order properties of the input process~\cite{Ogunfunmi2011}:
\begin{equation}
	\text{MSE} = \sigma_v^2 + \frac{\mu \sigma_v^2}{2-\mu} \Ex\left[\frac{1}{||\x_i||^2}\right] \tr(\Ree_x)	 (1+ 2\Gamma). \label{eq:MSE2}
\end{equation}
In this equation:
\begin{equation}
	\Gamma = \sum_{i=1}^M p_i \sum_{q=1}^{K-1} \left[\left(1- \frac{(1-p_i)^q}{1- (1-p_i)^{K}}\right) (1-\mu)^q\right], \label{eq:MSE3}
\end{equation}
and $p_i = \lambda_i /\tr(\Ree_x)$, where $\lambda_i$ are the eigenvalues of $\Ree_x$. 

In Fig \ref{fig:MSE512} we plot the MSE for $M=512$ and different inputs and values of $K$. The plots where obtained by averaging $J=400$ runs of the MSE to convergence and averaging the last 2000 samples of the steady-state value. The behavior of the MSE is weakly dependent on the correlation of the input signal for a small value of $\mu$. As $\mu$ increases the MSE increases, faster for more correlated signals, and with increasing $K$.
The proposed expression (\ref{eq:diagEeeH}) is accurate in the regime $\mu\ll 1$ as expected, and it captures the dependence with $K$ in this regime. However, since it is independent of the input process statistics, it cannot follow the MSE as $\mu$ increases. On the other hand, (\ref{eq:MSE}) is a poor approximation since it does not even depend on $K$. Finally, it is interesting to observe that (\ref{eq:MSE2}) requires the computation of the eigenvalues of the correlation matrix and other statistics which depend on the input signal. However, the simulations show that the end result depends weakly on the correlation of the input signals, even when $\mu \to 1$, where the color of the signal affects the MSE strongly. The prediction of  (\ref{eq:diagEeeH}) is almost the same as (\ref{eq:MSE2}), while 
(\ref{eq:diagEeeH}) is very simple to compute.
\begin{figure}
	\centering
	{\includegraphics[width=.5\linewidth, trim=.5cm 0.5cm  0.5cm 0.5cm,clip]{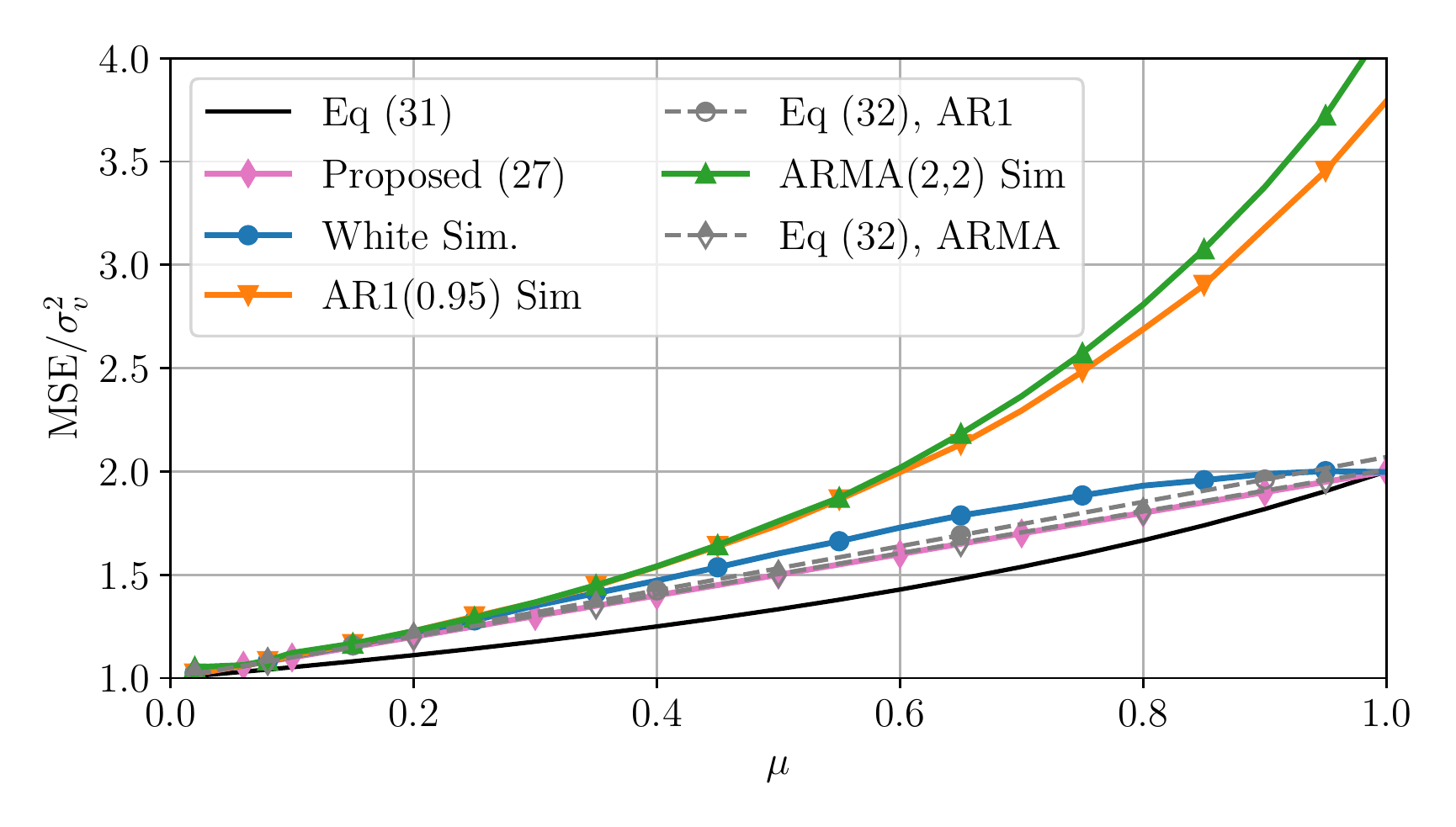}
	\subcaption{$K=2$.}}
	{\includegraphics[width=.5\linewidth, trim=.5cm 0.5cm  0.5cm 0.5cm,clip]{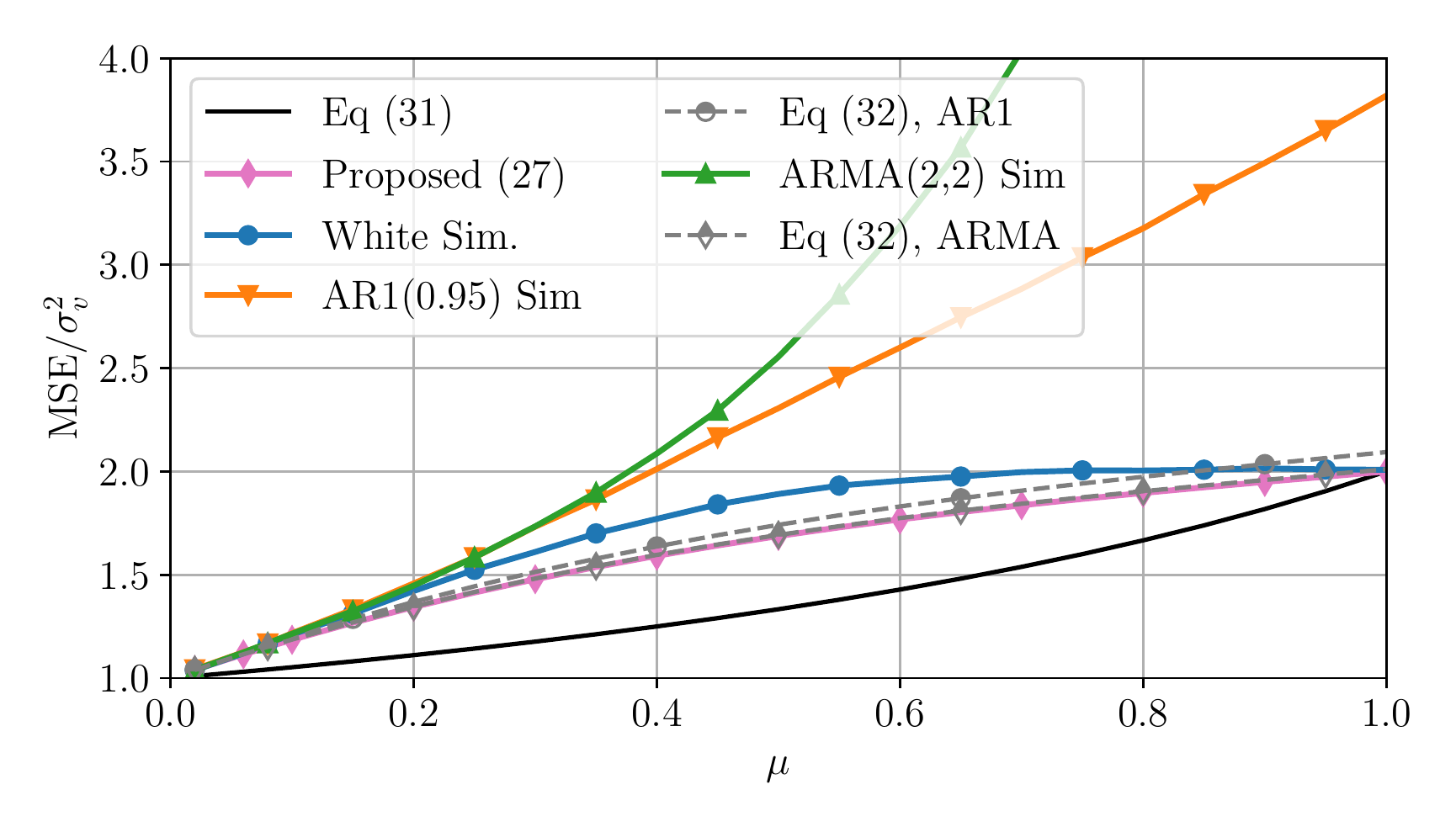}
	\subcaption{$K=4$.}}
	{\includegraphics[width=.5\linewidth, trim=.5cm 0.5cm  0.5cm 0.5cm,clip]{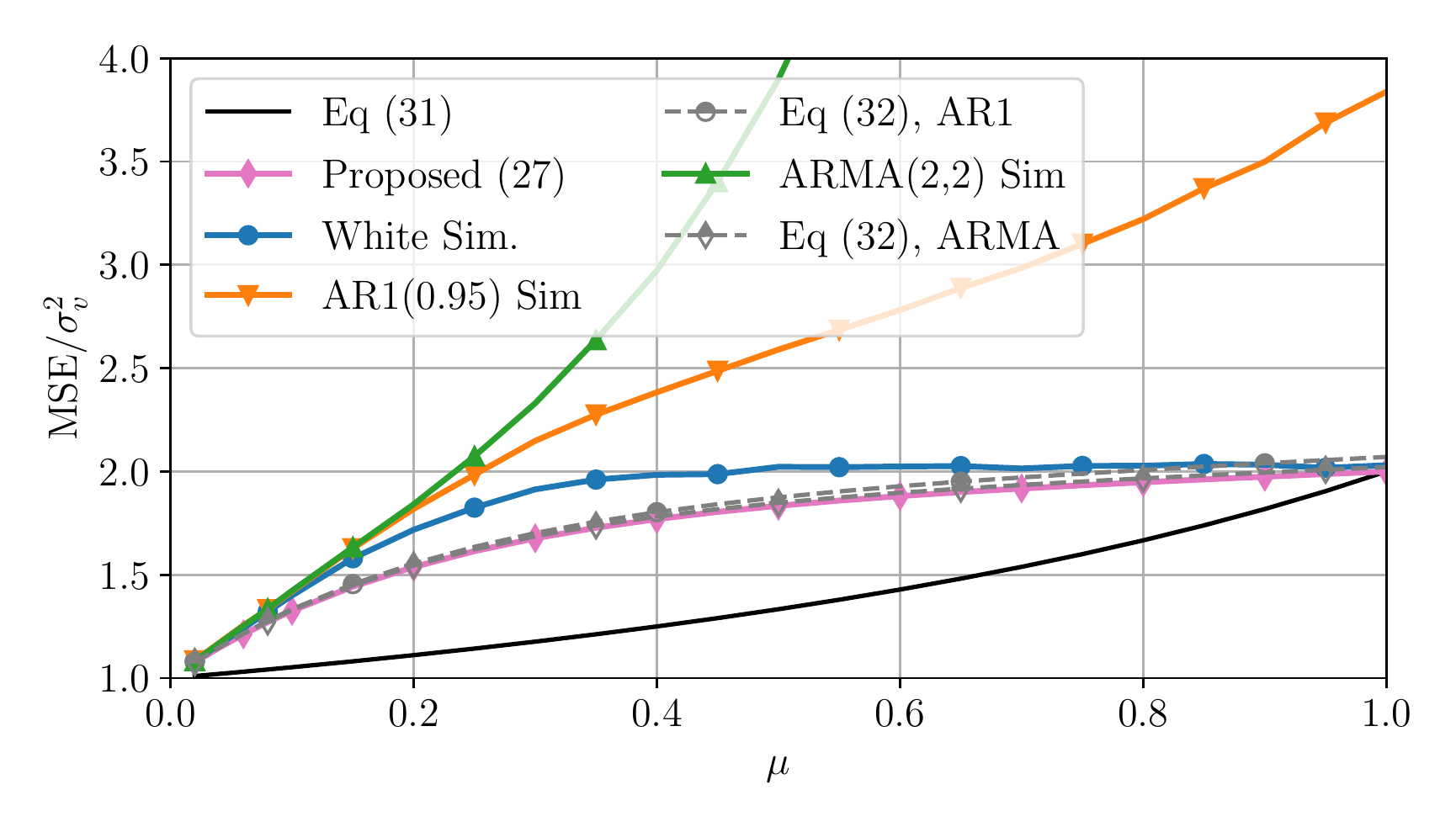}
	\subcaption{$K=8$.}}
\caption{Steady-state MSE for different inputs and values of $K$. The theoretical expression is given by (\ref{eq:diagEeeH}) for $q=1$. $M=512$. $\text{SNR}=30$dB. $J=400$.}
	\label{fig:MSE512}
\end{figure}

Finally, we explore the validity of (\ref{eq:diagEeeH}) to estimate the energy of the other components of the error vector. For this, we consider a white or AR1(0.95) input. The results for the ARMA(2,2) process are almost identical to those of the AR1(0.95). In Fig. \ref{fig:diagEeeH} we plot the estimated and simulated components of the energy of the errors for $K=4$, with $M=512$ and $M=32$. We see that the estimations are accurate for small $\mu$, which is were the approximations are obtained, and even for larger values of $\mu$ than expected. 

\begin{figure}
	\centering
	{\includegraphics[width=.5\linewidth, trim=.5cm 0.5cm  0.5cm 0.5cm,clip]{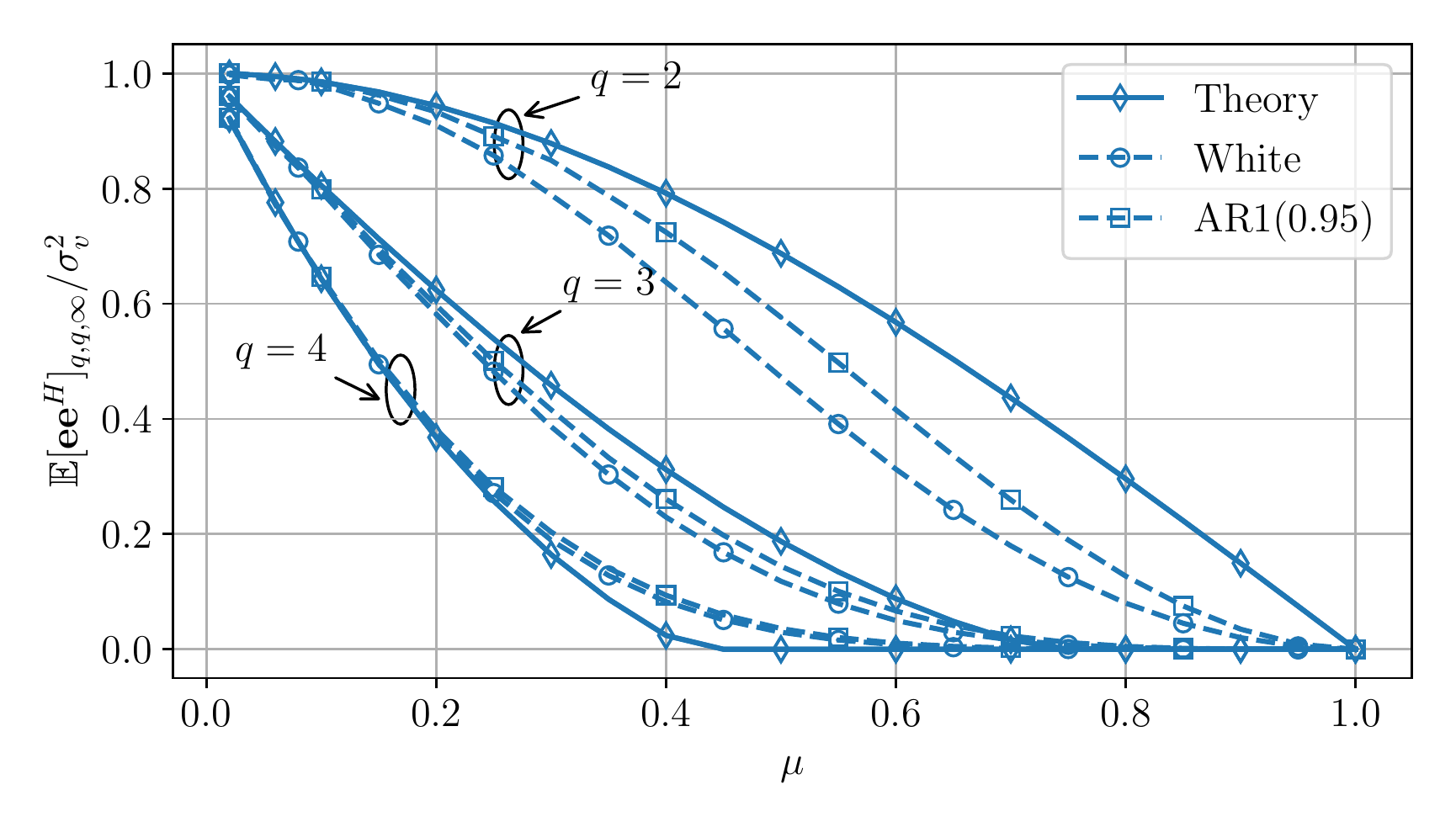}
		\subcaption{$M=512$.}}
	{\includegraphics[width=.5\linewidth, trim=.5cm 0.5cm  0.5cm 0.5cm,clip]{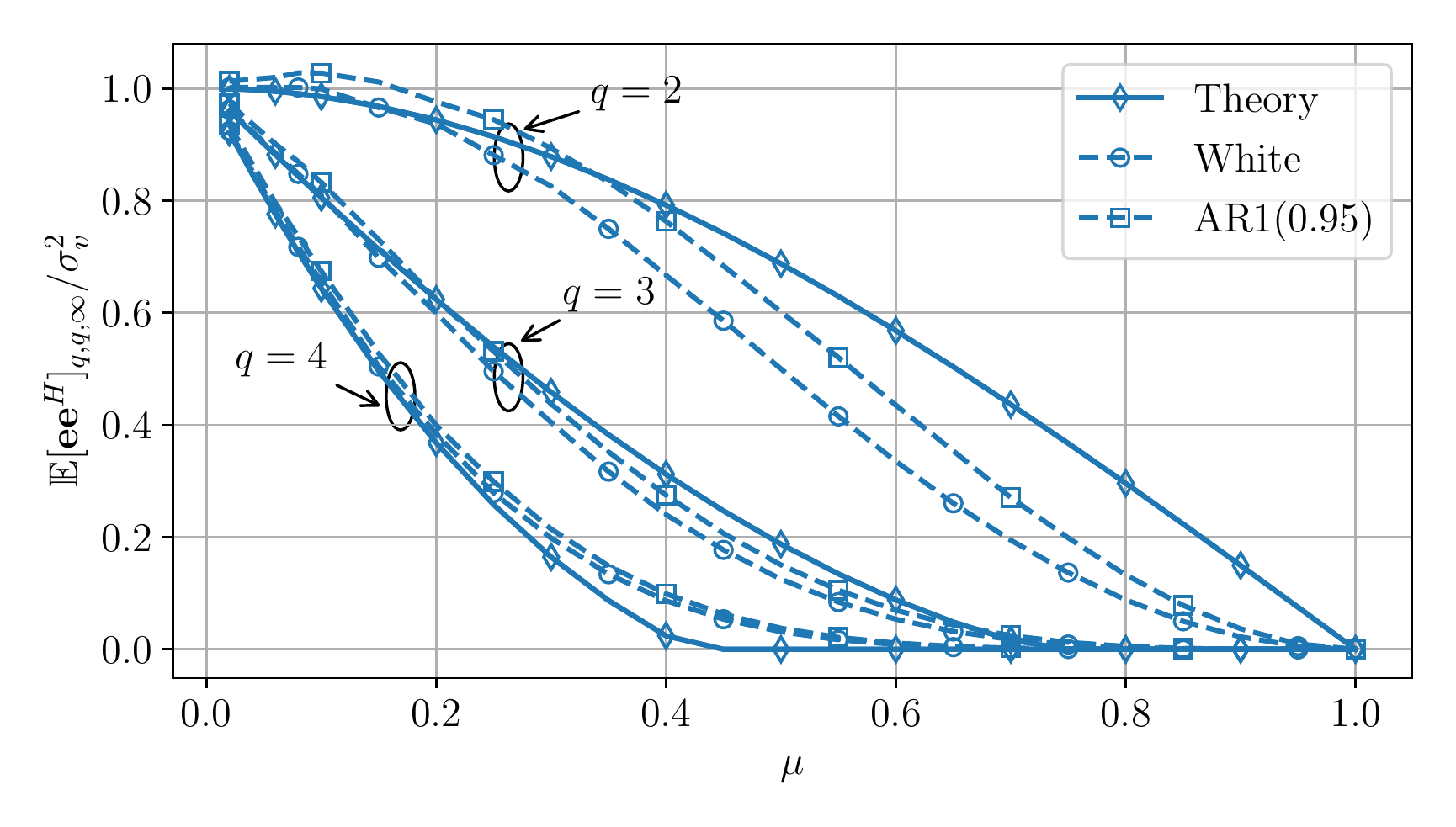}
		\subcaption{$M=32$.}}
	\caption{Steady-state energy of the error components. The theoretical expression is given by (\ref{eq:diagEeeH}). $\text{SNR}=30$dB. $J=240$.}
	\label{fig:diagEeeH}
\end{figure}

\section{Conclusion} \label{sec:closing}
In this paper we perfomed an analysis of some elements of the correlation matrix between the noise and \emph{a priori} error vector under mild assumptions. This analysis was done in the transient and steady-state of the algorithm. We applied this analysis to derive a new analysis for the steady state behavior of the algorithm in the small step-size regime. We have provided  closed-form formulas for the MSE of the algorithm which are simpler and with comparable accuracy than existing formulas which are more complex. Also we have provided approximate expressions for the energy of the other components of the error vector.  Using this initial correlation analysis it might be possible to draw further conclusions on the remaining elements of the correlation matrix. This could lead to a new steady-state characterization for correlated signals, and possibly, to a new transient state analysis as well. 

\section*{Acknowledgment}
The authors would like to thank Dr. Hernan Rey for the insightful discussions during the course of this work.

\appendices

\section{Proof of Theorem \ref{teo:main}}
We first present two subsections with preliminary results, and the proof itself is in Appendix \ref{sec:proofTEO}.
\subsection{Auxiliary linear algebra results}
We present some auxiliary definitions and properties  which will simplify the main proof of the theorem. The proofs in this section are straightforward and are omitted.
\begin{definition}
Let $\mathcal{F}_m$ be the set of matrices in $\C^{K\times K}$ whose last $K-m$ rows are zero:	
\begin{equation}
\mathcal{F}_m := \left\{\mathbf{A} \in \C^{K \times K} : \left[\mathbf{A}\right]_{i,j} = 0 \ \text{if $m < i \leq K$} \right\}.
\end{equation}
\end{definition}
\begin{definition} 
Let $\mathcal{C}_m$ be the set of matrices in $\C^{K\times K}$ whose first $m$ columns are zero:
\begin{equation}
	\mathcal{C}_m := \left\{\mathbf{A} \in \C^{K \times K} : \left[\mathbf{A}\right]_{i,j} = 0 \ \text{if $1 \leq j \leq m$}\right\}.
\end{equation}
\end{definition}
The following properties on the products of matrices hold:
\begin{propiedad}\label{PropF}
	If $\mathbf{A} \in \mathcal{F}_m$ and  $\Ree \in \C^{K \times K}$, then $\mathbf{A} \Ree \in \mathcal{F}_m$. 
\end{propiedad}
\begin{propiedad}\label{PropC}
	If $\mathbf{A} \in \mathcal{C}_m$ and $\Ree \in \C^{K \times K}$ then $ \Ree \mathbf{A}\in \mathcal{C}_m$. 
\end{propiedad}
\begin{propiedad}	\label{PropFC}
	If $\mathbf{A} \in \efe_n$, $\mathbf{B} \in \ce_m$ and $\Ree \in \C^{K \times K}$, then from Properties \ref{PropF} and \ref{PropC} we have that $\mathbf{A} \Ree \mathbf{B} \in \efe_n \cap \ce_m$. 
\end{propiedad}
\begin{propiedad}\label{CorrI}
	If $\mathbf{A} \in \efe_m$ and $j \in \mathbb{Z}$, $j < 0$ then $\Imonio_{K,j} \mathbf{A} \in \efe_{m-j}$, with $\Imonio_{K,j}$ given by (\ref{eq:Imonio}).
\end{propiedad}
\begin{propiedad}\label{diag0}
	If $\Ree \in \efe_m \cap \ce_m$ then $\Ree$ is an upper-triangular matrix, and $\diag\left(\Ree\right) = \mathbf{0}$. 
\end{propiedad}
\begin{definition} Let $\Id_{K,m}$ be a diagonal matrix such that the first $m$ elements in the diagonal are 0, and the rest are 1:
\begin{equation}
\left[\Id_{K,m}\right]_{i,j} = \begin{cases}  1 & \mbox{ if } i=j \mbox{ and } i \gt m \\ 0 & \mbox{otherwise}. \end{cases}
\end{equation}
Notice that $\Id_{K,0}$ is the $K\times K$ identity matrix.
\end{definition}
Consider the matrices $\Imonio_{K,m}$ given by (\ref{eq:Imonio}). When multiplying a matrix $\Imonio_{K, m}$ with another matrix it will move the rows of a matrix up or down and pad with zeros. In addition, $\Imonio_{K,0}$ is the $K\times K$ identity matrix.
We will use the following properties:
\begin{propiedad} 	\label{Imonio+-}
	For $1 \leq m \leq K-1$,  $\Imonio_{K,-m} \Imonio_{K,m} = \Id_{K,m}$.
\end{propiedad}
\begin{propiedad}\label{ProdImonio}
	If $m,n \in \{1,...,K-1\}$ then $\Imonio_{K,m} \Imonio_{K,n} = \Imonio_{K,n+m}$.  
\end{propiedad}

%
%
%
%

\subsection{A matrix decomposition for the shifted correlation matrix}
\label{Sec.descomp}
Consider a complex second order stationary process $x$. 
We may define a time-shifted autocorrelation matrix  $ \Ree_{x,-m} =  \Ex\left[[x_i,...,x_{i-K+1}]^* [x_{i-m},...,x_{i-m-K+1}]^T\right] \in \C^{K\times K}$. Taking $m=0$ gives the standard autocorrelation matrix. 
\begin{lemma} \label{lema:decomp}
For $ 0 < m < K$ the matrix $\Ree_{x,-m}$  can be written in terms of $\Ree_x$ as:
\begin{equation}
	\Ree_{x,-m} = \Imonio_{K,-m} \Ree_x + \M_m,\label{descompR}
\end{equation}
where $\M_m$ is some matrix in $\efe_m$. When $x$ is a white process, $\M_m$ is the zero matrix.
\end{lemma}
The proof is straightforward and is omitted. Essentially $\Imonio_{K,-m}$ moves the rows of $\Ree_x$ downwards padding with rows of zeros from the top, while $\M_m$ completes the first $m$ rows of $\Ree_{x,-m}$.

%
%
%
%
%

\subsection{Expansion of a product of matrices as a sum} \label{app:expand}
We now rewrite the product 
\begin{equation}
	\prod_{k=1}^{j-1} (\Id-\Ge_{k}) = (\Id - \Ge_1)(\Id - \Ge_2) ... (\Id - \Ge_{j-1}),
\end{equation}
which appears in (\ref{eavwhite2}) as a summation. 
Each of its terms can be obtained by selecting from each $(\Id - \Ge_k)$ either the identity  or  the  $\Ge_k$ matrix. Since the identity matrices do not affect the product, we can represent each term in the summation as a list of ordered indexes which correspond to the $\Ge_k$ matrices which appear in the summation term. Additionally, each term of the summation will be multiplied by $(-1)$ if there is an odd number of $\Ge_k$ matrices involved. 
The $p$-th term will be represented by an ordered list of indexes:
\begin{equation}
\mathcal{I}_p = [ \sigma_{1,p},...,\sigma_{\#I_p,p}],
\end{equation}
and the matrix product will be written as:
\begin{align}
\hspace{-1mm}	\prod_{k=1}^{j-1} (\Id-\Ge_{k}) &= \Id +\sum_p (-1)^{\#\mathcal{I}_p} \prod_{ n = 1}^{\#I_p} \Ge_{i-\sigma_{n,p}}
	\\ &= \Id +\sum_p (-1)^{\#I_p} \prod_{ n = 1}^{\#\mathcal{I}_p} \X_{\sigma_{n,p}} \ese_{\sigma_{n,p}}\X_{\sigma_{n,p}}^H. \label{eq:expansion}
\end{align}
The summation will contain $2^{j-1}-1$ terms and any indexed list $\mathcal{I}_p$ will satisfy the following properties:
\begin{enumerate}
	\item $\sigma_{i,p} \in \N \ \  \forall i>K$.
	\item $ \# \mathcal{I}_p < j$.
	\item	$i > \sigma_{1,p} > \sigma_{2,p} > ... > \sigma_{ \# \mathcal{I}_p,p} > i-j$.
\end{enumerate}

\subsection{Proof of Theorem \ref{teo:main}} \label{sec:proofTEO}
In order to prove the theorem we need to start from (\ref{eavwhite2}).  We focus on $K>2$, and leave $K=2$ for the reader. 
Using (\ref{eq:expansion}) in (\ref{eavwhite2}) and rearranging the terms we obtain ($K>2$):
\begin{multline} \small
	\Ex\left[\e_{a,i}\ve_i^H \right]  =  - \sigma_v^2 \sum_{j=1}^{K-1}\Ex\left[\X_i^H \J_j\right] \Imonio_{K,j}\\
	- \sigma_v^2\sum_{j=2}^{K-1} \sum_p (-1)^{\#\mathcal{I}_p}\Ex\left[\X_i^H \left(\prod_{\sigma_n \in \mathcal{I}_p}  \X_{\sigma_n} \ese_{\sigma_n}\X_{\sigma_n}^H \right) \J_j \right]  \Imonio_{K,j}.  \label{eavwhite5}
\end{multline}
The coefficients $\sigma_{n}$ are a function of $p$ also, but to simplify the notation, we do not make this dependence explicit, that is, $\sigma_{n,p} \equiv \sigma_n$. Likewise we do not include the limits of the summation in $p$ because it is not used. Now we simplify both terms of these expressions using the preliminary results of this appendix. 

The first term on the right side of (\ref{eavwhite5}) can be written as:
\begin{equation}
\hspace{-.5mm}	- \sigma_v^2 \sum_{j=1}^{K-1} \Ex\left[\X_i^H \J_j \right] \Imonio_{K,j} \hspace{-.5mm}	\approx \hspace{-.5mm} - \sigma_v^2 \sum_{j=1}^{K-1}\mu_{i-j} \Ree_{x,-j} \Ree_x^{-1} \Imonio_{K,j}  \hspace{-5mm}\label{aprox1}
\end{equation}
using (\ref{eq:approxR}). Now we have the following lemma:
\begin{lemma} \label{lema:exp1}
The right side of (\ref{aprox1})  can be written as:
\begin{equation} \label{eq:term1} 
\hspace{-.5mm}	\sum_{j=1}^{K-1} \hspace{-.5mm}\mu_{i-j} \Ree_{x,-j} \Ree_x^{-1} \Imonio_{K,j} \hspace{-.5mm} = \hspace{-.5mm}\sum_{j=1}^{K-1} \hspace{-.5mm} \mu_{i-j} \hspace{-.5mm} \left(\Id_{K,j} + \M_j \Ree_x^{-1} \Imonio_{K,j}\right)
\end{equation}
for  certain matrices $\M_j \in \efe_j$. In addition, $\M_j \Ree_x^{-1} \Imonio_{K,j} \in \ce_j \cap \efe_j$.
\end{lemma}
\begin{proof}
 To show this, we use  Lemma \ref{lema:decomp}, to write $\Ree_{x,-j} = \Imonio_{K,-j} \Ree_x + \M_j$, where $\M_j \in \efe_j$. After replacing this in (\ref{aprox1}) we use Property \ref{Imonio+-} to show that $\Id_{K,j} = \Imonio_{K,-j} \Imonio_{K,j}$. Finally, notice that using Property \ref{PropFC} we know that $\M_j \Ree_x^{-1} \Imonio_{K,j} \in \ce_j \cap \efe_j$.
\end{proof}

Then for the second term on the right side of (\ref{eavwhite5}) we replace $\ese_{\sigma_n}$ and $\J_j$ in terms of $\mu_{\sigma_n}$, $\mu_{i-j}$, $\X_{\sigma_n}$ and $\X_{i-j}$.
We then define: $c(p) = \# \mathcal{I}_p$ and define $\tilde{\mathcal{I}}(p)= [\sigma_1,...,\sigma_{c(p)-1}]$ as the list in which the last element has been removed.
Then we use (\ref{eq:approxR}) to approximate:
\begin{itemize}
	\item $\X_i^H \X_{\sigma_1} \approx M \Ree_{x,\sigma_1-i}$.
	\item $ \X_{\sigma_n}^H \X_{\sigma_{n+1}} \approx M \Ree_{x,\sigma_{n+1} - \sigma_{n}}$.
	\item $\X_{\sigma_{c(p)}}^H \X_{i-j} \approx M \Ree_{x,i-j-\sigma_{c(p)}}$.
\end{itemize}
With this we can approximate:
\begin{multline} 
	\displaystyle \Ex\left[\X_i^H \left(\prod_{\sigma_n \in \mathcal{I}_p} \hspace{-2mm} \X_{\sigma_n} (\X_{\sigma_n}^H\X_{\sigma_n})^{-1}\X_{\sigma_n}^H\right) \J_j \right] \Imonio_{K,j}  \approx
	\Ree_{x,\sigma_1-i} \\ \hspace{-5mm} \times \Ree_x^{-1}\hspace{-1mm} \left[\prod_{ \sigma_n \in \tildI } \hspace{-2mm} \Ree_{x,\sigma_{n+1}-\sigma_n} \Ree_x^{-1}\right] \hspace{-1mm}\Ree_{x,i-j-\sigma_{c(p)}}\Ree_x^{-1} \Imonio_{K,j}. \hspace{-3mm} \label{aprox2}
\end{multline}
Although this expression is involved, it is possible to prove the following Lemma:
\begin{lemma} \label{lema:exp2}
	The left side of (\ref{aprox2}) can be written as:
\begin{equation}
		\Ree_{x,\sigma_1-i} \Ree_x^{-1} \left(\prod_{ \sigma_n \in \tildI} \Ree_{x,\sigma_{n+1}-\sigma_n} \Ree_x^{-1}\right)  \Ree_{x,i-j-\sigma_{c(p)}} \Ree_x^{-1} \Imonio_{K,j}
		= \Id_{K,j} + \te_{0,p,j}, \label{aprox4}
	\end{equation}
	where $\te_{0,p,j} \in \efe_j \cap \ce_j$\footnote{We write $\te_{0,p,j}$ to indicate they depend on the $j$-th  and $p$-th indexes of the summations in (\ref{eavwhite5}). During the proof we do not explicit the indexes $p,j$.}. For a white input, $\te_{0,p,j}$ is the null matrix.
	
\end{lemma}

\begin{proof}
The proof proceeds by evaluating the product in (\ref{aprox4}) from right to left. To do  this we proceed by  induction, by showing that the following formula is valid:
\begin{equation}
		\left(\prod_{ \sigma_k \in \tildI: k \geq n} \Ree_{x,\sigma_{k+1}-\sigma_k} \Ree_x^{-1}\right) \Ree_{x,i-j-\sigma_{c(p)}}\Ree_x^{-1} \Imonio_{K,j}= \Imonio_{i-j-\sigma_{n}} \Imonio_{K,j} + \te_n, \label{inductiva}
\end{equation}
with $\te_n \in \efe_{t_n}\cap \ce_j $, $t_n = \max\left\{\sigma_n - (i-j),\sigma_n - \sigma_m\right\}$ and $n = 1\ldots (m-1)$.  The induction proceeds backwards, starting from $k = c(p) -1$ to $k = 1$.
With this expression, (\ref{aprox4}) can be simplified to complete the proof.

First we factor the matrix 	$\Ree_{x,i-j-\sigma_{c(p)}}$ using Lemma  \ref{lema:decomp}.
Then we can simplify the left side of (\ref{inductiva}) as:
\begin{equation}
	\Ree_{x,i-j-\sigma_{c(p)}}\Ree_x^{-1} \Imonio_{K,j} = \left(\Imonio_{i-j-\sigma_{c(p)}} \Imonio_{K,j}+  \mathbf{T}_{c(p)}\right), \label{eq:lasterm}
\end{equation}
where $\mathbf{T}_{c(p)}= \M_{c(p)}  \Ree_x^{-1}\Imonio_{K,j}$. Since $\M_{c(p)} \in \efe_{\sigma_{c(p)}-(i-j)}$ and $\Imonio_{K,j} \in \ce_j$  then using Property \ref{PropFC} we have $\mathbf{T}_{c(p)} \in \efe_{\sigma_{c(p)}-(i-j)} \cap \ce_j$.

Now prove that the result is valid for $n = {c(p)}-1$. In this case, the product (\ref{inductiva}) indexed by $k$ only has the term $k=c(p)-1$. Applying the decomposition from Lemma \ref{lema:decomp}  to this term and using (\ref{eq:lasterm}) we have:
\begin{multline}
	\left(\Ree_{x,\sigma_{c(p)}-\sigma_{c(p)-1}} \Ree_x^{-1}\right) \Ree_{x,i-j-\sigma_{c(p)}}\Ree_x^{-1} \Imonio_{K,j}  = \left(\Imonio_{\sigma_{c(p)}-\sigma_{c(p)-1}} + \M_{c(p)-1} \Ree_x^{-1}\right) \\ \times \left(\Imonio_{i-j-\sigma_{c(p)}} \Imonio_{K,j}+ \mathbf{T}_{c(p)}\right). \label{aprox6a}
\end{multline}
We now expand an analyze the four  terms of (\ref{aprox6a}):
\begin{itemize}
	\item Since  $\sigma_{c(p)} - \sigma_{c(p)-1} < 0$ and $i-j-\sigma_{c(p)}< 0$ we apply Property \ref{ProdImonio} to show that $\Imonio_{\sigma_{c(p)}-\sigma_{c(p)-1}} \Imonio_{i-j-\sigma_{c(p)}} \Imonio_{K,j}= \Imonio_{i-j-\sigma_{c(p)-1}} \Imonio_{K,j}.$
	
	\item Since $\Te_{c(p)} \in \efe_{\sigma_{c(p)}-(i-j)} \cap \ce_j$ we apply Property \ref{CorrI} to show that $\Imonio_{\sigma_{c(p)}-\sigma_{{c(p)}-1}}\mathbf{T}_{c(p)} \in \efe_{\sigma_{c(p)-1}-(i-j)} \cap \ce_j$. 

	\item $\M_{{c(p)}-1} \Ree_x^{-1} \Imonio_{i-j-\sigma_{c(p)}} \Imonio_{K,j}  \in \efe_{\sigma_{c(p)-1}-\sigma_c(p)} \cap \ce_j$ since $\M_{c(p)-1} \in \efe_{\sigma_{m}-\sigma_{m-1}}$ and $\Id_{K,j} \in \ce_j$ (Property \ref{PropFC}).
	
	\item $\M_{c(p)-1} \Ree_x^{-1} \mathbf{T}_{c(p)} \in \efe_{\sigma_{c(p)-1}-\sigma_c(p)} \cap \ce_j$ because $\M_{c(p)-1} \in \efe_{\sigma_{c(p)}-\sigma_{c(p)-1}}$ and $\te_{c(p)} \in \ce_j$ (Property \ref{PropFC}).
\end{itemize}
Thus, all the terms, except the first belong either to  $(\efe_{\sigma_{c(p)-1}-\sigma_c(p)} \cap \ce_j)$ or $(\efe_{\sigma_{c(p)-1}-(i-j)} \cap \ce_j)$. We now define $t_{c(p)-1} := \max\left\{\sigma_{c(p)-1}-(i-j),\sigma_{c(p)-1}-\sigma_{c(p)}\right\}$ and:
\begin{equation}
	\te_{c(p)-1} := \Imonio_{\sigma_{c(p)}-\sigma_{c(p)-1}}\mathbf{T}_{c(p)} +\M_{c(p)-1} \Ree_x^{-1} \Imonio_{i-j-\sigma_{c(p)}} \Imonio_{K,j} + \M_{c(p)-1} \Ree_x^{-1} \mathbf{T}_{c(p)},
\end{equation}
so that $\te_{c(p)-1} \in \efe_{t_{c(p)-1}} \cap \ce_j$. Then, (\ref{aprox6a}) is written as:
\begin{equation}
	\Ree_{x,\sigma_{c(p)}-\sigma_{c(p)-1}}  \Ree_x^{-1} \Ree_{x,i-j-\sigma_{c(p)}}\Ree_x^{-1}\Imonio_{K,j} = \Imonio_{i-j-\sigma_{c(p)-1}} \Imonio_{K,j} + \mathbf{T}_{c(p)-1}, \label{aprox7}
\end{equation}
where $\mathbf{T}_{c(p)-1} \in \efe_{t_{c(p)-1}} \cap \ce_j$ and $t_{c(p)-1}= \max\left\{\sigma_{c(p)-1}-(i-j),\sigma_{c(p)-1}-\sigma_m\right\}>0$, which shows that the expression is valid for $n=c(p)-1$.

Now we assume that  (\ref{inductiva}) is valid for $k = n+1$ and show that it is valid for $k=n$.
This means that:
\begin{multline}
	\prod_{ \sigma_k \in \tildI: k \geq n} \left(\Imonio_{\sigma_{k+1}-\sigma_k} + \M_{ k} \Ree_x^{-1} \right) \Ree_{x,i-j-\sigma_{c(p)}}\Ree_x^{-1} \Imonio_{K,j} \\=  \left(\Imonio_{\sigma_{n+1}-\sigma_n} + \M_{n} \Ree_x^{-1} \right) \left(\Imonio_{i-j-\sigma_{n+1}} \Imonio_{K,j} + \te_{n+1}  \right),
\end{multline} 
where we have applied the inductive hypothesis. We have that $\te_{n+1} \in \efe_{t_{n+1}} \cap \ce_j$ with $t_{n+1} =\max\left\{\sigma_{n+1}-(i-j),\sigma_{n+1}-\sigma_{c(p )}\right\}$. Expanding the product and analyzing the terms we have:
\begin{itemize}
	\item $\Imonio_{\sigma_{n+1}-\sigma_n}\Imonio_{i-j-\sigma_{n+1}} \Imonio_{K,j} = \Imonio_{i-j-\sigma_{n}} \Imonio_{K,j}$.
	\item $\Imonio_{\sigma_{n+1}-\sigma_n} \te_{t_{n+1}} \in \efe_{t_{n+1} + (\sigma_n-\sigma_{n+1})} \cap \ce_j$.
	\item $\M_{n} \Ree_x^{-1} \Imonio_{i-j-\sigma_{n+1}} \Imonio_{K,j} \in \efe_{\sigma_{n+1}-\sigma_n} \cap \ce_j$.
	\item $\M_{n} \Ree_x^{-1} \te_{t_{n+1}} \in \efe_{\sigma_{n+1}-\sigma_n} \cap \ce_j$.
\end{itemize}
But since $t_{n+1}>0$ we have $\sigma_n - \sigma_{n+1}+ t_{n+1}> \sigma_n-\sigma_{n+1}.$
In addition:
	\begin{align*}
		t_n &= \max\left\{\sigma_{n+1} - (i-j),\sigma_{n+1} - \sigma_m\right\} + \sigma_n - \sigma_{n+1} \\ 
		&= \max\left\{\sigma_n - (i-j),\sigma_n - \sigma_m\right\}.
	\end{align*}
So we find that:
\begin{equation}
	\te_n = \Imonio_{\sigma_{n+1}-\sigma_n} \te_{t_{n+1}} + \M_{n} \Ree_x^{-1} \Imonio_{i-j-\sigma_{n+1}} \Imonio_{K,j} +
	 \M_{n} \Ree_x^{-1} \te_{t_{n+1}} \in \efe_{t_n}\cap \ce_j,
\end{equation} 
with $t_n = \max\left\{\sigma_n - (i-j),\sigma_n - \sigma_m\right\}$. This proofs that (\ref{inductiva}) is valid.

We can now simplify (\ref{aprox4}) by replacing (\ref{inductiva}) with $k=1$:
\begin{multline}
	\Ree_{x,\sigma_1-i} \Ree_x^{-1} \left(\prod_{ \sigma_n \in \tildI} \Ree_{x,\sigma_{n+1}-\sigma_n} \Ree_x^{-1}\right) \\ \times \Ree_{x,i-j-\sigma_{c(p)}} \Ree_x^{-1} \Imonio_{K,j}
	= \Ree_{x,\sigma_1-i} \Ree_x^{-1} \left(\Imonio_{i-j-\sigma_1} \Imonio_{K,j} + \te_1\right).
\end{multline}
Applying once more the decomposition of Lemma \ref{lema:decomp} to $\Ree_{x,\sigma_1-i}$ we get: 
\begin{equation}
\Ree_{x,\sigma_1-i} \Ree_x^{-1} \left(\Imonio_{i-j-\sigma_1} \Imonio_{K,j} + \te_1\right)  = \Imonio_{\sigma_1-i} \Imonio_{i-j-\sigma_1} \Imonio_{K,j} + \te_0,
\end{equation}
with $\te_0 = \Imonio_{\sigma_1-i} \te_1 + \M_0 \Ree_x^{-1} \Imonio_{i-j-\sigma_1} \Imonio_{K,j} +\M_0 \Ree_x^{-1} \te_1$.

Following a similar reasoning as with the other terms we can prove that $\te_0 \in \efe_{t_0} \cap \ce_j$ with  $t_0 = \max \left\{i-\sigma_1, i-\sigma_1+\max\left\{\sigma_1 - (i-j),\sigma_1 - \sigma_m\right\}\right\} = i-\sigma_1+\max\left\{\sigma_1 - (i-j),\sigma_1 - \sigma_m\right\} =\max\left\{j,i - \sigma_m\right\}$.
But since $\sigma_m > i-j$ we have that $t_0 = j$, which shows that $\te_0 \in \efe_j \cap \ce_j$. Finally, we conclude the proof by noting that: 
$	\Imonio_{\sigma_1-i} \Imonio_{i-j-\sigma_1} \Imonio_{K,j} =  \Imonio_{K,-j}\Imonio_{K,j} = \Id_{K,j},$
which is obtained using Properties \ref{Imonio+-} and \ref{ProdImonio}. The result for a white input follows by noting that all the $\M$ matrices from the decomposition (\ref{descompR}) are zero for a white input.
\end{proof}
To conclude the proof of  Theorem \ref{teo:main}, we use the results of Lemmas \ref{lema:exp1} and \ref{lema:exp2}.
We replace (\ref{aprox4}) in (\ref{aprox2}) and this result together with (\ref{eq:term1}) in (\ref{eavwhite5}) to obtain:
\begin{multline}
	\Ex\left[\e_{a,i}\ve_i^H \right] \approx - \sigma_v^2 \sum_{j=1}^{K-1}\mu_{i-j} \left(\Id	_{K,j} + \M_j \Ree_x^{-1} \Imonio_{K,j}\right) \\
	-\sigma_v^2\left[\sum_{j=2}^{K-1} \sum_p(-1)^{f(p)} \mu_{i-j} \left(\prod_{\sigma_n \in \mathcal{I}_p}\mu_{\sigma_n}\right) \left(\Id_{K,j} + \te_{0,p,j}\right) \right]. 
\end{multline}
Using Property \ref{diag0} we have that $\te_{0,p,j}$ and $\M_j \Ree_x^{-1} \Imonio_{K,j}$ are  upper triangular matrices, so $\Ex\left[\e_{a,i}\ve_i^H \right]$ also is.
In particular for a white input, $\te_{0,p,j}$ and $\M_j$ are the zero matrix, so $\Ex\left[\e_{a,i}\ve_i^H \right]$ is a diagonal matrix. In both cases, from Property \ref{diag0} we have that $\diag(\te_{0,p,j}) = \diag(\M_j \Ree_x^{-1} \Imonio_{K,j}) = \mathbf{0}$ for all $j, p$. So we can simplify this expression to obtain:
\begin{multline}
	\diag \Ex\left[\e_{a,i}\ve_i^H \right] \approx   - \sigma_v^2 \mu_{i-1} \diag\left(\Id_{K,1}\right) \\
	-\sigma_v^2\sum_{j=2}^{K-1} \mu_{i-j}\left[1+ \sum_p(-1)^{f(p)}  \left(\prod_{\sigma_n \in \mathcal{I}_p}\mu_{\sigma_n}\right)\right] \diag\left(\Id_{K,j} \right)
	.\label{aprox11}
\end{multline}
Now we conclude the proof by observing that 
\begin{equation}
	1+\sum_p (-1)^{f(p)} \prod_{\sigma_n \in \mathcal{I}_p}\mu_{\sigma_n} = \prod_{k=1}^{j-1}\ (1-\mu_{i-k}). \label{aprox12}
\end{equation}
that is, we revert the matrix factorization from Appendix \ref{Sec.descomp}
 which is also valid for scalars.  $ \Ex\left[\e_{a,i}\ve_i^H \right]_{q,q}$ can be found by considering the individual elements of the diagonal.


\ifCLASSOPTIONcaptionsoff
  \newpage
\fi



%

%
\bibliographystyle{IEEEtran}
\bibliography{IEEEabrv,biblio}








\end{document}